\documentclass[sigplan,screen]{acmart}
\makeatletter
\providecommand{\mdseries@tt}{}
\makeatother

\usepackage{algorithmic}
\usepackage{alltt}%
\usepackage{amsmath,amssymb,amsfonts,amsthm}
\usepackage{comment}
\usepackage{enumitem}
\usepackage{etextools}
\usepackage{hyphenat}
\usepackage{ifthen}
\usepackage{mathpartir}
\usepackage{mathtools}
\usepackage{minibox}
\usepackage{multirow}
\usepackage{stmaryrd}
\usepackage{suffix}
\usepackage{supertabular}
\usepackage{textcomp}
\usepackage{thm-restate}
\usepackage[normalem]{ulem}
\usepackage{xcolor}
\usepackage{xspace}
\definecolor{linkcolor}{rgb}{0.65,0,0}
\definecolor{citecolor}{rgb}{0,0.65,0}
\definecolor{urlcolor}{rgb}{0,0,0.65}
\usepackage{breakurl}
\usepackage{cleveref}
\usepackage[frozencache]{minted}
\usepackage{halloweenmath}

\newcommand{\ccode}[1]{\mintinline[fontsize=\small]{c}{#1}}
\newminted[cblock]{c}{xleftmargin=\parindent, linenos=none, breaklines=true, fontsize=\small}
\newcommand{\fcode}[1]{\mintinline[fontsize=\small]{fact_lexer.py:FaCTLexer -x}{#1}}
\newminted[fblock]{fact_lexer.py:FaCTLexer -x}{xleftmargin=\parindent, linenos=none, breaklines=true, fontsize=\small}

\makeatletter
\def\my@overarrow@#1#2#3{\vbox{\ialign{##\crcr #1#2\crcr \noalign{\kern-\p@\nointerlineskip}$\m@th \hfil #2#3\hfil $\crcr}}}
\newcommand{\amsvectb}{%
  \mathpalette {\my@overarrow@\vectfillb@}}
\newcommand{\vecbar}{%
  \scalebox{0.6}[.8]{$\relbar$}}
\def\vectfillb@{\arrowfill@\vecbar\vecbar{\raisebox{-4.2pt}[\p@][\p@]{\hspace{.1pt}$\mathord\mathchar"017E$}}}
\makeatother

  {\begin{list}{$\blacktriangleright$}%
    {\leftmargin=\parindent \itemsep=2pt \topsep=2pt
     \parsep=0pt \partopsep=0pt}}%
  {\end{list}}

\renewcommand\vec{\amsvectb}

\newcommand{\tool}{Pitchfork\xspace}

\newcommand{\ignore}[1]{}

\newcommand\ie{i.e.,\xspace}
\newcommand\eg{e.g.,\xspace}

\definecolor{lightblue}{rgb}{.8,.9,1}
\definecolor{darkred}{rgb}{.8,.0,0}
\newcommand\pcnum[1]{\uline{#1}}
\newcommand\idxnum[1]{\overline{#1}}

\makeatletter
    \newcommand{\DeclarePairedDelimiterCase}[2]{%
        \newcommand#1[1][]{%
            \ifthenelse{\equal{##1}{normal}}%
            {#2}%
            {%
                \ifthenelse{\equal{##1}{big}\OR\equal{##1}{Big}\OR\equal{##1}{bigg}\OR\equal{##1}{Bigg}}%
                {\expandnext{#2[}{\csname##1\endcsname}]}%
                {#2*}%
            }%
        }%
    }
    \newcommand{\DeclarePairedDelimiterY}[4][Temp]{%
        \expandafter\DeclarePairedDelimiter\csname#2#1\endcsname{#3}{#4}%
        \expandnext{\expandnext{\DeclarePairedDelimiterCase}{\csname#2\endcsname}}{\csname#2#1\endcsname}%
    }
    \newcommand{\DeclarePairedDelimiterXY}[6][Temp]{%
        \expandafter\DeclarePairedDelimiterX\csname#2#1\endcsname[#3]{#4}{#5}{#6}%
        \expandnext{\expandnext{\DeclarePairedDelimiterCase}{\csname#2\endcsname}}{\csname#2#1\endcsname}%
    }
\makeatother

\newcommand{\SONI}{SCT\xspace}

\newcommand{\mypara}[1]{\smallskip\noindent\emph{\textbf{{#1.}}}}

\newcommand{\RegSet}{\ensuremath{\mathcal{ R}}}
 
 \newcommand{\V}{\ensuremath{\mathcal{ V}}} 
\newcommand{\Labels}{\ensuremath{\mathcal{ L}}} 
  
\newcommand{\Confs}{{\sf Confs}}

\newcommand{\BuffInstr}{\ensuremath{\textsf{TransInstr}}}

\DeclarePairedDelimiterY{floor}{\lfloor}{\rfloor}
\DeclarePairedDelimiterY{ceil}{\lceil}{\rceil}
\DeclarePairedDelimiterY{abs}{\lvert}{\rvert}

\newcommand\hasrpc[2]{\,\vdash^{\!\!\!#1}_{\!\!\!#2}\,}
\WithSuffix\newcommand\hasrpc*{\hasrpc\rp\pc}
\newcommand\aug{\mapsto}

\renewcommand\implies{\Rightarrow}
\newcommand\stuck\bot
\newcommand\join{\sqcup}

\renewcommand\ll{\llbracket}
\newcommand\rr{\rrbracket}
\newcommand\step[2]{\xhookrightarrow[#1]{#2}}

\newcommand\bigstep[3][N]{\,{}_{#3}\!{\Downarrow}_{#2}^{#1}\,}
\newcommand\bigstepnon[2]{\,{}_{#1}\!{\Downarrow}_{#2}\,}

\newcommand\freshi{\textsf{MAX}(\buf) + 1} 

\newcommand\retirable[1]{\textsf{MIN}(\buf)=#1}
\newcommand\regmap[1]{(\buf +_i \reg)(#1) }
\newcommand\regmapname{(\buf +_i \reg) }

\newcommand\specstep[2][]{(\reg, \mem, \n, \buf) \xhookrightarrow[#2\,]{#1} }

\newcommand\n{n}
\newcommand\reg{\rho}

\newcommand\mem{\mu}
\newcommand\buf{\textit{buf}\,}
\newcommand\ktrue{\textsf{true}}
\newcommand\kfalse{\textsf{false}}
\newcommand\rv{r\!v}
\newcommand\rets{\sigma}
\newcommand\rsp{r_{\mathit{sp}}}
\newcommand\vsp{v_\l}
\newcommand\rret{r_{\mathit{tmp}}}

\renewcommand{\l}{\ell}

\newcommand\rf{r}
\renewcommand\a{a}
\newcommand\rmi[1]{#1 \setminus #1(i)}

\newcommand\seq{\mathit{seq}}
\newcommand\loweqv{\simeq_\kpub}

\newcommand\public{\textsf{public}}%
\newcommand\secret{\textsf{secret}}%
\newcommand\kpub{\textsf{pub}}%
\newcommand\ksec{\textsf{sec}}%

\newcommand\load{\textsf{load}}
\newcommand\store{\textsf{store}}
\newcommand\op{\textsf{op}}
\newcommand\cond{\textsf{br}}
\newcommand\jump{\textsf{jump}}
\newcommand\lfence{\textsf{fence}}
\newcommand\call{\textsf{call}}

\newcommand\ret{\textsf{ret}}
\newcommand\jmpi{\textsf{jmpi}}

\newcommand\aload[3]{(#1 = \load(#2,#3))}      %
\newcommand\astore[3]{\store(#1,#2,#3)}        %
\newcommand\aopx[4]{#2 = \op(#1,#3,#4)}       %
\newcommand\aop[4]{(#2 = \op(#1,#3,#4))}       %
\newcommand\acond[4]{\cond(#1,#2,#3,#4)}       %
\newcommand\alfence[1]{\lfence\ #1}            %
\newcommand\acall[2]{\call(#1,#2)}             %
\newcommand\aret{\ret}                         %
\newcommand\ajmpi[1]{\jmpi(#1)}                %

\newcommand\iloadx[2]{#1 = \load(#2)}         %
\newcommand\iload[2]{(#1 = \load(#2))}         %
\newcommand\ivalx[2]{#1 = #2}                 %
\newcommand\ival[2]{(#1 = #2)}                 %
\newcommand\istore[2]{\store(#1,#2)}           %
\newcommand\iopx[3]{#2 = \op(#1,#3)}          %
\newcommand\iop[3]{(#2 = \op(#1,#3))}          %
\newcommand\icond[5]{\cond(#1,#2,#3,(#4,#5))}  %
\newcommand\ijmp[1]{\jump\ #1}                 %
\newcommand\ilfence{\lfence}                   %
\newcommand\icall{\call}                       %
\newcommand\iret{\ret}                         %
\newcommand\ijmpi[2]{\jmpi(#1,#2)}             %

\newcommand\addr{\textit{addr}}
\newcommand\addrf[1]{\opx{\addr}{#1}}
\newcommand\opcode{\textit{op}}
\newcommand\opf[1]{\ll \opcode(#1) \rr}
\newcommand\opx[2]{\ll #1(#2) \rr}
\newcommand\condop{\textit{op}}

\newcommand\push{\textit{push}}
\newcommand\pop{\textit{pop}}
\newcommand\sttop{\textit{top}}

\newcommand\fetchx{\textsf{fetch}}
\newcommand\fetch{\fetchx}
\newcommand\fetchg[1]{\textsf{fetch:}\,\,#1}
\newcommand\executex{\textsf{execute}}
\newcommand\execute[1]{\executex~#1}
\newcommand\executei{\executex~i}
\newcommand\executevalue[1]{\execute{#1:\,\textsf{value}}}
\newcommand\executeaddr[1]{\execute{#1:\,\textsf{addr}}}
\newcommand\executeivalue{\executevalue{i}}
\newcommand\executeiaddr{\executeaddr{i}}
\newcommand\executeg[2]{\executex~#1:\,\textsf{fwd}\,#2}
\newcommand\executeig[1]{\executex~i:\,\textsf{fwd}\,#1}
\newcommand\retirex{\textsf{retire}}
\newcommand\retire{\retirex}

\newcommand\oread{\texttt{read}~}
\newcommand\owrite{\texttt{write}~}
\newcommand\ojump{\texttt{jump}~}
\newcommand\ofwd{\texttt{fwd}~}
\newcommand\omiss{\texttt{rollback}}

\newcommand\ofwdx{\texttt{fwd}}

\newcommand{\later}[1]{\colorbox{lightgray}{$#1$}}

\newcommand{\ucsdmark}{{$^{\dagger}$}}
\newcommand{\inriamark}{{$^{\star}$}}
\newcommand{\mpimark}{{$^{\spadesuit}$}}
\newcommand{\imdeamark}{{$^{\clubsuit}$}}

\newcommand{\jointauthor}[1]{\author{#1}}
\newcommand{\jointaffiliation}[1]{\affiliation{#1}}

\jointauthor{%
  Sunjay~Cauligi\ucsdmark\quad
  Craig~Disselkoen\ucsdmark\quad
  Klaus~v.~Gleissenthall\ucsdmark
}

\jointauthor{%
  Dean~Tullsen\ucsdmark\quad
  Deian~Stefan\ucsdmark\quad
  Tamara~Rezk\inriamark\quad
  Gilles~Barthe\mpimark\imdeamark
}
\jointaffiliation{\vspace{.5em}
  \ucsdmark{UC San Diego, USA}\qquad
  \inriamark{INRIA Sophia Antipolis, France}\qquad \\
  \mpimark{MPI for Security and Privacy, Germany}\qquad
  \imdeamark{IMDEA Software Institute, Spain}
  \vspace{.25em}
}

\renewcommand{\author}[1]{}
\renewcommand{\affiliation}[2][]{}
\renewcommand{\institution}[1]{}
\renewcommand{\country}[1]{}

\renewcommand{\authors}{Sunjay Cauligi, Craig Disselkoen, Klaus v.~Gleissenthall, Dean Tullsen, Deian Stefan, Tamara Rezk, and Gilles Barthe}
\renewcommand{\shortauthors}{S.~Cauligi, C.~Disselkoen, K.~v.~Gleissenthall, D.~Tullsen, D.~Stefan, T.~Rezk, and G.~Barthe}

\setcopyright{none}
\acmPrice{}
\acmDOI{}
\acmYear{}
\copyrightyear{}
\acmSubmissionID{}
\acmISBN{}
\acmConference[]{}{}{}
\acmBooktitle{}

\begin{document}
\title{Constant-Time Foundations for the New Spectre Era}

\begin{abstract}
  The constant-time discipline is a software-based countermeasure used
  for protecting high assurance cryptographic implementations against
  timing side-channel attacks. Constant-time is effective (it protects
  against many known attacks), rigorous (it can be formalized using
  program semantics), and amenable to automated verification. Yet, the
  advent of micro-architectural attacks makes constant-time as it
  exists today far less useful.

  This paper lays foundations for constant-time programming in the presence of
  speculative and out-of-order execution.  We present an operational semantics
  and a formal definition of constant-time programs in this extended setting.
  Our semantics eschews formalization of microarchitectural features (that are
  instead assumed under adversary control), and yields a notion of
  constant-time that retains the elegance and tractability of the usual notion.
  We demonstrate the relevance of our semantics in two ways:  First, by
  contrasting existing Spectre-like attacks with our definition of
  constant-time.  Second, by implementing a static analysis tool,
  \tool, which detects violations of our extended constant-time property in
  real world cryptographic libraries.
\end{abstract}



\maketitle

\section{Introduction}

Protecting secrets in software is hard.
Security and cryptography engineers must write programs that protect secrets,
both at the source level and when they execute on real hardware.
Unfortunately, hardware too easily divulges information about a program's
execution via \emph{timing side-channels}---e.g., an attacker can learn secrets
by simply observing (via timing) the effects of a program on the hardware
cache~\cite{ge2018survey}.

The most robust way to deal with timing side-channels is via
\emph{constant-time programming}---the paradigm used to implement almost all
modern cryptography~\cite{bearSSL, libsodium_repo, mbedtls_repo, cttk,
codingRules}.
Constant-time programs can neither branch on secrets nor access memory
based on secret data.\footnote{More generally, constant-time programs
  cannot use secret data as input to any variable-time operation---e.g.,
  floating-point multiplication.  }
These restrictions ensure that programs do not leak secrets
via timing side-channels on hardware {\it without} microarchitectural
features.

Unfortunately, these guarantees are moot for most modern hardware:
Spectre~\cite{spectre-v1}, Meltdown~\cite{meltdown},
ZombieLoad~\cite{zombieload}, RIDL~\cite{ridl}, and Fallout~\cite{fallout} are
all dramatic examples of attacks that exploit microarchitectural features.
These attacks reveal that code that is deemed constant-time in the usual sense
may, in fact, leak information on processors with microarchitectural features.
The decade-old constant-time recipes are no longer enough.\footnote{OpenSSL
found this situation so hopeless that they recently updated their security
model to explicitly exclude ``physical system side
channels''~\cite{opensslsecuritypolicymay19}.}

In this work, we lay the foundations for constant-time in the presence
of microarchitectural features that have been exploited in recent
attacks: out-of-order and speculative execution. We focus on
constant-time for two key reasons.
First, \emph{impact}: constant-time programming is largely used in
real-world crypto libraries---and high-assurance code---where
developers already go to great lengths to eliminate leaks via
side-channels.
Second, \emph{foundations:} constant-time programming is already
rooted in foundations, with well-defined semantics~\cite{barthe2014system,cauligi:2019:fact}.
These semantics
consider very powerful attackers---e.g., attackers
in~\cite{barthe2014system} have control over the cache and the
scheduler. An advantage of considering powerful attackers is that the
semantics can overlook many hardware details---e.g., since the cache
is adversarially controlled, there is no point in modeling it
precisely---making constant-time amenable to automated verification
and enforcement.

\mypara{Contributions}
We first define a semantics for an abstract, three-stage (fetch,
execute, and retire) machine. Our machine supports out-of-order and
speculative execution by modeling \emph{reorder buffers} and
\emph{transient instructions}, respectively. We assume
that attackers have complete control over microarchitectural features
(e.g., the branch target predictor) when executing a victim program and
model the attacker's control over predictors using \emph{directives}.
This keeps our semantics simple yet powerful: our semantics abstracts over all
predictors when proving security---of course, assuming that predictors
themselves do not leak secrets.
We further show how our semantics can be extended to capture new
predictors---e.g., a hypothetical \emph{memory aliasing} predictor.

We then define \emph{speculative constant-time}, an extension of
constant-time for machines with out-of-order and speculative
execution. 
This definition allows us to discover microarchitectural side channels in a
principled way---all four classes of Spectre attacks as classified by Canella
et al.~\cite{canella18systematic}, for example, manifest as violations of our
constant-time property.

We further use our semantics as the basis for a prototype analysis tool, \tool, built on top of the
\texttt{angr} symbolic execution engine~\cite{shoshitaishvili2016state}.
Like other symbolic analysis tools, \tool suffers from path explosion, which
limits the depth of speculation we can analyze.
Nevertheless, we are able to use \tool to detect multiple Spectre bugs in real code.
We use \tool to detect leaks in the well-known Kocher test
cases~\cite{kocher-tests} for Spectre v1, as well as our more extensive test
suite which includes Spectre v1.1 variants.
More significantly, we use \tool to analyze---and find leaks in---real
cryptographic code from the libsodium, OpenSSL, and curve25519-donna libraries.

\mypara{Open source}
\tool and our test suites are open source and available at
\url{https://pitchfork.programming.systems}.

\section{Motivating Examples}
\label{sec:motivating}
In this section, we show why classical constant-time programming is
insufficient when attackers can exploit microarchitectural features.
We do this via two example attacks and show how these attacks are captured by
our semantics.

\mypara{Classical constant time is not enough}
Our first example consists of 3 lines of code, shown in \Cref{fig:v1-attack}
(top right). The program, a variant of the
classical Spectre v1 attack~\cite{spectre-v1}, branches on the value
of register $r_a$ (line $\pcnum 1$). If $r_a$'s value is smaller than
$4$, the program jumps to program location $\pcnum 2$, where it
uses $r_a$ to index into a public array~$A$, saves the value into
register~$r_b$, and uses~$r_b$ to index into another public array~$B$.
If $r_a$ is larger than or equal to $4$ (\ie the index is out of
bounds), the program skips the two load instructions and jumps to
location~$\pcnum 4$.
In a sequential execution, this program neither loads nor
branches on secret values. It thus trivially satisfies the constant\hyp{}time
discipline.

However, modern processors do not execute sequentially.
Instead, they continue fetching instructions before prior instructions
are complete.
In particular, a processor may continue fetching instructions
beyond a conditional branch, before evaluating the branch condition.
In that case, the processor \emph{guesses} which branch will be taken.
For example, the processor may erroneously guess that the branch condition at
line $\pcnum 1$ evaluates to $\ktrue$, even though~$r_a$ contains value~$9$.
It will therefore continue down the ``true'' branch speculatively.
In hardware, such guesses are made by a branch prediction unit, which may
have been mistrained by an adversary.

These guesses, as well as additional choices such as execution order, are
directly supplied by the adversary in our semantics.
We model this through a series of \emph{directives}, as shown on the bottom
left of \Cref{fig:v1-attack}.
The directive $ \fetchg \ktrue $ instructs our model to speculatively follow
the $\ktrue$ branch and to place the fetched instruction at index~$\idxnum 1$
in the \emph{reorder buffer}.
Similarly, the two following $\fetch$ directives place the loads at indices
$\idxnum 2$ and $\idxnum 3$ in the buffer.
The instructions in the reorder buffer, called \emph{transient instructions},
do not necessarily match the original instructions, but can contain
additional information (see Table~\ref{tab:instructions}).
For instance, the transient version of the branch instruction records which
branch has been speculatively taken.

In our example, the attacker next instructs the model to execute the
first load, using the directive $\execute \idxnum 2$.
Because the bounds check has not yet been executed, the load reads from the
secret element~$\textit{Key}[1]$, placing the value in $r_b$.
The attacker then issues directive $\execute \idxnum 3$ to execute the
following $\load$; this $\load$'s address is calculated as $\mathtt{44} +
\textit{Key}[1]$.
Accessing this address affects externally visible cache state, allowing the
attacker to recover $\textit{Key}[1]$ through a cache
side-channel attack~\cite{ge2018survey}.
This is encoded by the leakage observation shown in red on the bottom right.
Though this secret leakage cannot happen under sequential execution, our
semantics clearly highlights the possible leak when we account for
microarchitectural features.

\begin{figure}
  \small
  \centering

  \begin{tabular}{llcrl}

    \multicolumn{2}{c}{Registers}                    && \multicolumn{2}{c}{Program} \\
    $~r $               & $ \reg(r) $                && $n $        & $ \mem(n) $ \\
    \cline{1-2} \cline {4-5}
    $r_a$               & $ 9_\kpub $                && $\pcnum 1 $ & $ \acond{\texttt{>}}{(4, r_a)}{\pcnum 2}{\pcnum 4} $ \\
    \multicolumn{2}{c}{Memory}                       && $\pcnum 2 $ & $ \aload{r_b}{[\mathtt{40}, r_a]}{\pcnum 3} $ \\
    $~a $               & $ \mem(a) $                && $\pcnum 3 $ & $ \aload{r_c}{[\mathtt{44}, r_b]}{\pcnum 4} $ \\
    \cline{1-2}
    $ \mathtt{40..43} $ & $ \textit{array A}_\kpub $ && $\pcnum 4 $ & \dots \\
    $ \mathtt{44..47} $ & $ \textit{array B}_\kpub $ &&             & \\
    $ \mathtt{48..4B} $ & $ \textit{array Key}_\ksec $ \\

  \end{tabular}
  \vspace{1.25em}\\
  Speculative execution: \\
  \begin{tabular}{lll}

    Directive              & Effect on reorder buffer     & Leakage \\\hline
    $ \fetchg \ktrue $     & \vphantom{\LARGE{hack}}$ \idxnum 1 \mapsto \icond{\texttt{>}}{(4,r_a)}{\pcnum 2}{\pcnum 2}{\pcnum 4} $ & \\
    $ \fetch $             & $ \idxnum 2 \mapsto \iload{r_b}{[\mathtt{40}, r_a]} $                           & \\
    $ \fetch $             & $ \idxnum 3 \mapsto \iload{r_c}{[\mathtt{44}, r_b]} $                           & \\
    $ \execute \idxnum 2 $ & $ \idxnum 2 \mapsto \ival{r_b}{\textit{Key}[1]_\ksec} $                   & $ \oread \mathtt{49}_\kpub $ \\
    $ \execute \idxnum 3 $ & $ \idxnum 3 \mapsto \ival{r_c}{X} $                                             & $ \color{darkred}{\oread a_\ksec }$ \\

  \end{tabular}

  where $a = \textit{Key}[1]_\ksec + \mathtt{44}$
  \caption{
    Example demonstrating a Spectre~v1 attack.
    The branch at $\pcnum 1$ acts as bounds check for array~$A$.
    The execution speculatively ignores the bounds check, and
    leaks a byte of the secret~$Key$.
  }
\label{fig:v1-attack}
\end{figure}

\mypara{Modeling hypothetical attacks}
Next, we give an example of a hypothetical class of Spectre attack
captured by our extended semantics.
The attack is based on a microarchitectural feature which would allow
processors to speculate whether a store and load pair might
operate on the same address, and forward values between them~\cite{spoiler, store2leak}.

We demonstrate this attack in \Cref{fig:vnew-attack}.
The reorder buffer, after all instructions have been fetched, is shown in the
top right.
The program stores the value of register~$r_b$ into the
$\textit{secretKey}_\ksec$ array and eventually loads
two values from public arrays.
The attacker first issues the directive $\executevalue{\idxnum 2}$;
this results in a buffer where the store
instruction at $\idxnum 2$ has been modified to record the resolved
value $ x_\ksec$.
Next, the attacker issues the directive $\executeg{\idxnum 7}{\idxnum 2}$,
which causes the model to
mispredict that the load at $\idxnum 7$ aliases with the store at $\idxnum 2$,
and thus to forward the value $x_\ksec$ to the load.
The forwarded value $x_\ksec$ is then used in the address $a =48 + x_\ksec$
of the load instruction at index $\idxnum 8$.
There, the loaded value $X$ is irrelevant, but the address $a$ is leaked to the
attacker, allowing them to recover the secret value $x_\ksec$.  The
speculative execution continues and rolls back when the
misprediction is detected (details on this are given in
\Cref{sec:semantics}), but at this point, the secret has
already been leaked.

As with the example in \Cref{fig:v1-attack}, the program in this example follows the
(sequential) constant-time discipline, yet leaks during speculative execution.
But, both examples are insecure under our new notion of \emph{speculative
constant-time} as we discuss next.

\begin{figure}
  \small
  \centering

  \begin{tabular}{llccl}

    \multicolumn{2}{c}{Registers}                      && \multicolumn{2}{c}{Reorder buffer} \\
    $~r $               & $ \reg(r) $                  && $i$        & $ \buf(i) $ \\
    \cline{1-2} \cline {4-5}
    $r_a$               & $ 2_\kpub $                  && \vphantom{\LARGE{hack}}$\idxnum 2 $ & $ \istore{r_b}{[\mathtt{40}, r_a]} $ \\
    $r_b$               & $ x_\ksec $                  &&              & \dots \\
    \multicolumn{2}{c}{Memory}                         && $\idxnum 7 $ & $ \iload{r_c}{[\mathtt{45}]} $ \\
    $~a $               & $ \mem(a) $                  && $\idxnum 8 $ & $ \iload{r_c}{[\mathtt{48}, r_c]} $ \\
    \cline{1-2}
    $ \mathtt{40..43} $ & $ \textit{secretKey}_\ksec $ && \\
    $ \mathtt{44..47} $ & $ \textit{pubArrA}_\kpub $   && \\
    $ \mathtt{48..4B} $ & $ \textit{pubArrB}_\kpub $   && \\

  \end{tabular}
  \vspace{1.25em}\\
  Speculative execution\\ 
  \begin{tabular}{lll}

    Directive                                & Effect on $\buf$                                                     & Leakage \\\hline
    $ \executevalue{\idxnum 2} $             & \vphantom{\LARGE{hack}}$ \idxnum 2 \mapsto \istore{x_\ksec}{[\mathtt{40}, r_a]} $           & \\
    $ \executeg{\idxnum 7}{\idxnum 2} $      & $ \idxnum 7 \mapsto \iload{r_c}{[\mathtt{45}], x_\ksec, \idxnum 2} $ & \\
    $ \execute{\idxnum 8} $                  & $ \idxnum 8 \mapsto \ival{r_c}{X\{\bot,a\}} $                        & \color{darkred}{$\oread a_\ksec $} \\
    $ \executeaddr{\idxnum 2} $              & $ \idxnum 2 \mapsto \istore{r_b}{\mathtt{42}_\kpub} $                & $ \ofwd \mathtt{42}_\kpub $ \\
    \multirow{2}{*}{$ \execute{\idxnum 7} $} & \multirow{2}{*}{$ \{\idxnum 7,\idxnum 8\} \notin \buf $}             & $ \omiss, $ \\
                                             &                                                                      & $\ofwd \mathtt{45}_\kpub$ \\

  \end{tabular}

  where $a = x_\ksec + \mathtt{48}$
\caption{Example demonstrating a hypothetical attack abusing an aliasing predictor. %
  This attack differs from prior speculative data forwarding attacks in that branch misprediction is not needed.}
\label{fig:vnew-attack}
\end{figure}

\begin{table*}[ht!]
  \caption{Instructions and their transient instruction form.
}
\label{tab:instructions}
\begin{minipage}{1.0\textwidth}
  \begin{center}
   \begin{tabular}{rlll}
     &{\bf Instruction} & {\bf Transient form(s)} \\
     \toprule
     arithmetic operation & \multirow{2}{*}{$\aop{\opcode}{r}{\vec \rv}{n'}$}&$\iop{\opcode}{\rf}{\vec \rv}$ & \emph{(unresolved op)} \\
     (\opcode\ specifies opcode)&&$\ival{\rf}{v_\l}$             & \emph{(resolved value)}\\
     \midrule
     \multirow{2}{*}{conditional branch} &\multirow{2}{*}{$\acond{\condop}{\vec \rv}{n^\ktrue}{n^\kfalse}$}&$\icond{\condop}{\vec \rv}{n_0}{n^\ktrue}{n^\kfalse}$ & \emph{(unresolved conditional)} \\
     && $\ijmp{n_0}$                                          & \emph{(resolved conditional)}\\
    \midrule
     & \multirow{4}{*}{$\aload{r}{\vec \rv}{\n'} $}& $\iload{r}{\vec \rv}^n$ & \emph{(unresolved load)} \\
     memory load&& $\iload{\rf}{\vec \rv,( v_\l, j)}^n $ & \emph{(partially resolved load with dependency on $j$)} \\
     (at program point $n$)&& $\ival{\rf}{v _{\l}\{\bot, \a \}}^n$ &\emph{(resolved load without dependencies)} \\
     && $\ival{\rf}{v _{\l}\{j, \a \}}^n$ &\emph{(resolved load with dependency on $j$)} \\
    \midrule
     \multirow{2}{*}{memory store} & \multirow{2}{*}{$\astore{\rv}{\vec \rv}{\n'} $}& $\istore{\rv}{\vec \rv} $             & \emph{(unresolved store)} \\
     && $\istore{v_\l}{a_\l} $                & \emph{(resolved store)}\\
    \midrule
     \multirow{1}{*}{indirect jump} & \multirow{1}{*}{$\ajmpi{\vec \rv}$}& $\ijmpi{\vec \rv}{n_0}$ & \emph{(unresolved jump predicted to $n_0$)} \\
    \midrule
     \multirow{2}{*}{function calls} & \multirow{1}{*}{$\acall{n_f}{n_\mathit{ret}}$}& $\icall$ & \emph{(unresolved call)} \\
     & \multirow{1}{*}{$\aret$}& $\iret$ & \emph{(unresolved return)} \\
    \midrule
     \multirow{1}{*}{speculation fence} & \multirow{1}{*}{$\alfence \n$}& $\ilfence$ & \emph{(no resolution step)} \\
\bottomrule
     \end{tabular}
  \end{center}
\end{minipage}
\end{table*}

\section{Speculative Semantics and Security}\label{sec:semantics}

In this section we define the notion of  speculative constant time, and propose 
a speculative semantics that models execution on modern processors. 
We start by laying the groundwork for our definitions and semantics.

\mypara{Configurations}
A configuration $ C \in \Confs$ represents the state of
execution at a given step.
It is defined as a tuple $(\reg, \mem, \n, \buf)$ where:
\begin{itemize}
  \item $\reg: \RegSet \rightharpoonup \V$ is a map from a finite set of
    register names $\RegSet$ to values;
  \item $\mem: \V \rightharpoonup  \V $ is a memory;
  \item $\n: \V$ is the current program point;
  \item $\buf : \mathbb{N} \rightharpoonup  \BuffInstr$ is the reorder
    buffer.
    
  \end{itemize}
  
\mypara{Values and labels}
As a convention, we use $n$ for memory
addresses that map to instructions, and $\a$ for addresses that
map to data. 
Each value is annotated with a label from a lattice of security labels
with join operator $\sqcup$.  For brevity, we sometimes omit $\public$ label
annotation on values.

Using labels, we define an equivalence $\loweqv$ on
configurations. We say that two configurations are equivalent if they coincide
on $\public$ values in registers and memories.

\mypara{Reorder buffer}
The \emph{reorder buffer} maps buffer indices
(natural numbers) to transient instructions.
We write $\buf(i)$ to denote the instruction at index $i$ in buffer
$\buf$, if $i$ is in $\buf$'s domain. We write $\buf[i \mapsto
\textnormal{\underline{instr}}]$ to denote the result of extending
$\buf$ with the mapping from $i$ to $\textnormal{\underline{instr}}$,
and $\rmi{\buf}$ for the function formed by
removing $i$ from $\buf$'s domain. We write
$\buf[j : j < i]$ to denote the restriction of $\buf$'s domain to
all indices $j$, s.t. $j < i$ (\ie removing all mappings at indices
$i$ and greater).
Our rules add and remove indices in a way that ensures that $\buf$'s
domain will always be contiguous.

\mypara{Notation}
We let ${\sf MIN}(M)$ (resp. ${\sf MAX}(M)$)  denote the minimum
(maximum) index in the domain of a mapping $M$. 
We denote the empty mapping as
$\emptyset$ and let ${\sf MIN}(\emptyset)={\sf MAX}(\emptyset) =0$.

For a formula $\varphi$, we may discuss the bounded
highest (lowest) index for which a formula holds. We write
$ max(j) < i : \varphi(j)$ to mean that $j$ is the highest
index less than $i$ for which~$\varphi$ holds, and define $ min(j) >i
: \varphi(j)$ analogously.

\mypara{Register resolve function}
In \Cref{fig:reg-resolve}, 
 we  define the
  \emph{register resolve
function}, which we use to determine the value of a register in the presence of
transient instructions in the reorder buffer. 
For index~$i$ and register~$r$, the function may \textbf{(1)} return the
latest assignment to~$r$ prior to position $i$ in the buffer, if the
corresponding operation is already resolved; \textbf{(2)} return the value
from the register map $\reg$, if there are no pending assignments to $r$ in
the buffer; or \textbf{(3)} be undefined.
Note that if the latest assignment to $r$ is yet unresolved then
$\regmap{r} = \bot$.
\begin{figure}
\[
\regmap{r} =
  \begin{cases}
    v_\l  & \mathit{if} \  max(j) < i : \buf(j)=\ival{r}{\_} \land {}\\& \hspace{2em}\buf(j) = \ival{r}{v_\l} \\
    \reg(r) & \mathit{if} \ \forall j<i :  \buf(j) \neq \ival{r}{\_} \\
    \bot & \mathit{otherwise}
  \end{cases}
\]
  \caption{Definition of the register resolve function.}
  \label{fig:reg-resolve}
\end{figure}
We extend this definition to values by defining $\regmap{v_\l}=v_\l$
for all $v_\l \in \V$, and lift it to lists of registers or values
using a pointwise lifting.
\subsection{Speculative Constant-Time}
We present our new notion of constant-time security %
in terms of a small-step semantics, which
relates program configurations, observations, and attacker directives.

Our semantics does not directly model caches, nor
any of the predictors used by speculative semantics.
Rather, we model externally visible effects---memory accesses and control flow---by
producing a sequence of \emph{observations}.
We can thus reason about \emph{any} possible cache implementation, as any
cache eviction policy can be expressed as a function of
the sequence of observations.
Furthermore, exposing control flow observations directly in our semantics makes
it unnecessary for us to track various other side channels. Indeed, while
channels such as port contention or register renaming produce
distinct measurable effects~\cite{spectre-v1}, they only serve to leak the path
taken through the code---and thus modeling these observations separately would be
redundant.
For the same reason, we do not model a particular branch prediction
strategy; we instead let the attacker resolve scheduling non-determinism by
supplying a series of \emph{directives}.

This approach has two important consequences.
First, the use of observations and directives allows our semantics to remain
\emph{tractable} and \emph{amenable to verification}. For instance, we do not
need to model the behavior of the cache or any branch predictor.
Second, our notion of speculative constant-time is \emph{robust},
\ie it holds for all possible branch predictors and replacement
policies---assuming that they do not leak secrets directly, a
condition that is achieved by all practical hardware implementations.

Given an attacker directive~$d$, we use
$\smash{C \hspace{.2em}\mathrlap{\hookrightarrow}\hspace{.4em}{}_d^{\raisebox{.1em}{\scriptsize\hspace{.05em}$o$}}\hspace{.6em} C'}$
to denote the execution step from configuration $C$ to configuration $C'$ that
produces observation $o$.
Program execution is defined from the small-step semantics in the usual style.
We use $C \bigstep{D}{O} C'$ to denote a sequence of execution steps from $C$
to $C'$. Here $D$ and $O$ are the concatenation of the single-step directives
and leakages, respectively; $N$ is the number of retired instructions, i.e., $N
= \#\{d \in D \mathop{|} d = \retire\}$.  When such a big step from $C$ to $C'$
is possible, we say $D$ is a \emph{well-formed} schedule of directives for $C$.
We omit $D$, $N$, or $O$ when not used.

\begin{restatable}[Speculative constant-time]{definition}{defSpecONI}
  We say a configuration $C$ with schedule $D$ satisfies
  \emph{speculative constant-time} (SCT) with respect to a
  low-equivalence relation $\loweqv$ iff for every $C'$ such that
  $C \loweqv C'$:
  \[ C \bigstepnon{D}{O} C_1 \text{ iff } C' \bigstepnon{D}{O'} C_1'
    \text{ and } C_1 \loweqv C_1' \text{ and } O = O'. \]
  A program satisfies SCT iff every initial configuration satisfies SCT under
  any schedule.
\end{restatable}

\mypara{Aside, on sequential execution} Processors work hard to
create the illusion that assembly instructions are executed sequentially. We validate our
semantics by proving equivalence with respect to sequential execution. Formally, we define
\emph{sequential schedules} as schedules that execute and retire
instructions immediately upon fetching them. We attach to each
program a canonical sequential schedule and write $C \bigstep{\seq}{}
C'$ to model execution under this canonical schedule.
Our sequential validation is defined relative to an equivalence $\approx$
on configurations. Informally, two configurations are equivalent if
their memories and register files are equal, even if their speculative
states may be different.
\begin{restatable}[Sequential equivalence]{theorem}{thmSeqCorrectness}
\label{thm:seq-correctness}
Let  $C$  be an initial configuration and $D$ a well-formed schedule for $C$.
If $C \bigstep{D}{} C_1$, then $C \bigstep{\seq}{} C_2$ and $C_1 \approx C_2$.
\end{restatable}
Complete definitions, more properties, and proofs are given in \Cref{sec:full-proofs}.

\subsection{Overview of the Semantics} 

As shown in \Cref{tab:instructions},
each instruction has a \emph{physical} form and one or more \emph{transient} forms.
Our semantics operates on these instructions similar to a multi-stage processor pipeline.
Physical instructions are \emph{fetched} from memory and become
transient instructions in the reorder buffer. They are then \emph{executed} until they
are fully resolved. Finally they are \emph{retired}, updating the non-speculative state in the configuration.

In the rest of this section, we show how we
model speculative execution (\Cref{sec:cond}), memory operations
(\Cref{sec:forwarding}), aliasing prediction
(\Cref{sec:alias}), and fence instructions (\Cref{sec:lfence}). 
We also briefly describe indirect jumps and function calls (\Cref{sec:indirect}),
which are presented in full in \Cref{sec:extended-semantics}.

Our semantics captures a variety of existing Spectre variants, including v1 (\Cref{fig:v1-attack}), v1.1
(\Cref{fig:v1.1-attack}), and v4
(\Cref{fig:v4-attack}), as well as a new hypothetical variant (\Cref{fig:vnew-attack}).
Additional variants (\eg v2 and \emph{ret2spec}) can be expressed with the extended semantics given in \Cref{sec:extended-semantics}.
Our semantics shows that these attacks violate SCT by producing observations depending on secrets.

\subsection{Speculative Execution}
\label{sec:cond}
We start with the semantics for \emph{conditional branches} which
introduce speculative execution.

\mypara{Conditional branching}
The physical instruction for conditional branches has the form
$\acond{\condop}{\vec \rv}{n^\ktrue}{n^\kfalse}$,
 where $\condop$
is a Boolean operator whose result determines whether or not to
execute the jump, $\vec \rv$ are the operands to $\condop$, and $n^\ktrue$
and $n^\kfalse$ are the program points for the \emph{true} and
\emph{false} branches, respectively.

We show $\cond$'s transient counterparts in Table~\ref{tab:instructions}.
The unresolved form extends the physical instruction with a program point $n_0$, which is
used to record the branch that is executed ($n^\ktrue$ or $n^\kfalse$)
speculatively, and may or may not correspond to the branch that is actually taken once
$\condop$ is resolved. The resolved form contains the final jump target.

\mypara{Fetch}
We give the rule for the fetch stage below.
\begin{mathpar}
  \infer[cond-fetch]
  { \mem(\n) = \acond{\condop}{\vec \rv}{n^\ktrue}{n^\kfalse} \\
    i = \freshi \\\\
    \buf' = \buf[i \aug \icond{\condop}{\vec \rv}{n^\ktrue}{n^\ktrue}{n^\kfalse}] }
  { \specstep{\fetchg \ktrue} (\reg, \mem, \n^\ktrue, \buf') }
\end{mathpar}
The \textsc{cond-fetch} rule speculatively executes the branch determined
by a Boolean value $b$ given by the directive. We
show the case for $b=\ktrue$; the case for $\kfalse$ is analogous.
The rule updates the current program point~$n$, allowing
execution to continue along the specified branch. The rule then
records the chosen branch
$n^\ktrue$ (resp.~$n^\kfalse$) in the transient jump instruction.

This semantics models the behavior of most modern processors.
Since the target of the branch cannot be resolved in the
fetch stage, speculation allows execution to
continue and not stall until the branch target is resolved.
In hardware, a branch predictor chooses which branch to execute; in our
semantics, the directives $\fetchg \ktrue$ and $\fetchg \kfalse$ determine
which of the rules to execute.
This allows us to abstract over all possible predictor implementations.

\mypara{Execute}
Next, we describe the rules for the execute stage.
\begin{mathpar}
  \infer[cond-execute-correct]
  { \buf(i) = \icond{\condop}{\vec \rv}{n_0}{n^\ktrue}{n^\kfalse} \\
    \later{\forall j < i : \buf(j) \neq \lfence} \\
    \regmap{\vec \rv} = \vec{v_\l} \\
    \opf{\vec{v_\l}} = \ktrue_\l \\
    n^\ktrue = n_0 \\
   \buf' = \buf[i \aug \ijmp n^\ktrue] }
  { \specstep[\ojump n^\ktrue_\l]{\executei} (\reg, \mem, \n, \buf') }
\end{mathpar}\begin{mathpar}
  \infer[cond-execute-incorrect]
  { \buf(i) = \icond{\condop}{\vec \rv}{n_0}{n^\ktrue}{n^\kfalse} \\
    \later{\forall j < i : \buf(j) \neq \lfence} \\
    \regmap{\vec \rv} = \vec{v_\l} \\
    \opf{\vec{v_\l}} = \ktrue_\l \\
    n^\ktrue \neq n_0 \\
    \buf' = \buf[j : j < i][i \aug \ijmp{n^\ktrue}] }
  { \specstep[\omiss, \ojump n^\ktrue_\l]{\executei} (\reg, \mem, \n^\ktrue, \buf') }
\end{mathpar}
Both rules evaluate the condition $\condop$ via an evaluation
function $\ll \cdot \rr$. In both, the function produces $\ktrue$; but the
$\kfalse$ rules are analogous. The rules then compare the
actual branch target $n_\ktrue$ against the speculatively chosen
target~$n_0$ from the fetch stage.

If the \emph{correct} path was chosen during speculation, \ie $n_0$ agrees
with the correct branch $n^\ktrue$, rule \textsc{cond-execute-correct}
updates $\buf$ with the fully resolved $\jump$ instruction and
emits an observation: $\ojump n^\ktrue_\l$. This models
an attacker that can observe control flow, \eg by timing executions
along different paths.
The leaked observation $n^\ktrue$ has label $\l$, propagated from the evaluation
of the condition.

In case the \emph{wrong} path was taken during speculation, \ie the
calculated branch $n^\ktrue$ \emph{disagrees} with $n_0$, the semantics must
roll back all execution steps along the erroneous path.
For this, rule \textsc{cond-execute-incorrect} removes all entries in
$\buf$ that are newer than the current instruction (\ie all entries
$j \geq i$), sets the program point $\n$ to the correct branch, and
updates $\buf$ at index $i$ with correct value for the resolved
$\jump$ instruction.
Since an attacker can observer misspeculation through instruction
timing~\cite{spectre-v1}, the rule issues a $\omiss$ observation in addition
to the $\ojump$ observation.

\mypara{Retire}
The rule for the retire stage is shown below; its only effect
is to remove the $\jump$ instruction from the
buffer.
\begin{mathpar}
  \infer[jump-retire]
  { \retirable{i} \\
    \buf(i) = \ijmp{n_0} \\
    \buf' = \rmi{\buf} }
  { \specstep{\retire} (\reg, \mem, \n, \buf') }
\end{mathpar}

\begin{figure}
  \small
  (a) Predicted correctly \\
  \begin{tabular}{lll}
    $ i~ $        & Initial $ \buf(i) $                                              & $ \buf(i) $ after  $\textsf{execute}~\idxnum 4$\  \\\hline
    $ \idxnum 3 $ & $ \ival{r_b}{4} $                            & $\ival{r_b}{4} $ \\
    $ \idxnum 4 $ & $ \icond{\texttt{<}}{(2,r_a)}{\pcnum 9}{\pcnum 9}{\pcnum {12}} $ & $ \ijmp \pcnum 9 $ \\
    $ \idxnum 5 $ & $ \iop{\texttt{+}}{r_c}{(1,r_b)} $                            & $ \iop{\texttt{+}}{r_c}{(1,r_b)} $ \\
  \end{tabular}
      \vspace{1em} \\
  (b) Predicted incorrectly \\
  \begin{tabular}{lll}
    $ i~ $        & Initial $ \buf(i) $                                                 & $ \buf(i) $ after  $\textsf{execute}~\idxnum 4$\\\hline
    $ \idxnum 3 $ & $ \ival{r_b}{4} $                            & $ \ival{r_b}{4}$ \\
    $ \idxnum 4 $ & $ \icond{\texttt{<}}{(2,r_a)}{\pcnum {12}}{\pcnum 9}{\pcnum {12}} $ & $ \ijmp \pcnum 9 $ \\
    $ \idxnum 5 $ & $ \iop{\texttt{*}}{r_d}{(r_g,r_h)} $                            & - \\
  \end{tabular}

  \caption{
    Correct and incorrect branch prediction.
    Initially, $r_a=3$.
    In (b), the misprediction causes a rollback to $\idxnum 4$.
  }
  \label{fig:cond-misspec}
\end{figure}

\mypara{Examples}
\Cref{fig:cond-misspec} shows how branch prediction affects the
reorder buffer. In part (a), the branch at index $\idxnum 4$ is
predicted correctly. The jump instruction is resolved, and execution
proceeds as normal.
In part (b), the branch at index $\idxnum 4$ is incorrectly
predicted. Upon executing the branch, the misprediction is detected,
and $\buf$ is rolled back to index $\idxnum 4$.

\subsection{Memory Operations}
\label{sec:forwarding}

The physical instruction for loads is $\aload{r}{\vec \rv}{\n'}$,
while the form for stores is $\astore{\rv}{\vec \rv}{\n'}$.
As before, $\n'$ is the program point of the next instruction.
For $\load$s, $r$ is the register receiving the result; for $\store$s,
$\rv$ is the register or value to be stored.
For both $\load$s and $\store$s, $\vec \rv$ is a list of operands (registers
and values) which are used to calculate the operation's target
address.

Transient counterparts of $\load$ and $\store$ are given in Table~\ref{tab:instructions}. 
We annotate unresolved $\load$ instructions with the program point of the
physical instruction that generated them; we omit annotations whenever
  not used.
Unresolved and resolved $\store$ instructions share the same syntax,
but for resolved $\store$s, both
address and operand are required to be single values.

\mypara{Address calculation}
We assume an arithmetic operator $\addr$ which
calculates target addresses for $\store$s and $\load$s from its operands.
We leave this operation abstract in order to model a large variety of
architectures.
For example, in a simple
addressing mode, $\addrf{\vec v} $ might compute the sum of its
operands; in an x86-style address mode, $\addrf{[v_1,v_2,v_3]}$
might instead compute $v_1 + v_2 \cdot v_3$.

\mypara{Store forwarding}
Multiple transient $\load$ and $\store$ instructions may exist concurrently in
the reorder buffer.
In particular, there may be unresolved $\load$s and $\store$s that will read
or write to the same address in memory.
Under a naive model, we must wait to execute $\load$ instructions
until all prior $\store$ instructions have been retired, in case
they write to the address we will $\load$ from. Indeed, some real-world processors
behave exactly this way~\cite{chen_baer_1992}.

For performance, most modern processors implement \emph{store-forwarding} for
memory operations:
if a
$\load$ reads from the same address as a prior $\store$ and the $\store$ has
already been resolved, the
processor can \emph{forward} the resolved value to the $\load$. The $\load$ can then proceed without
waiting for the $\store$ to commit to
memory~\cite{wong14storetoload}.

To model these store forwarding semantics, we use
annotations to recall if a $\load$ was resolved from memory or forwarding.
A resolved $\load$ has the form $\ival{\rf}{v_\l{\{j, \a\}}}^n$,
where the index $j$ records either the buffer index of the $\store$
instruction that forwarded its value to the $\load$, or $\bot$ if
the value was taken from memory.
We also record the memory address $\a$ associated with the data,
and retain the program point $\n$ of the $\load$
instruction that generated the value instruction.
The resolved $\load$ otherwise behaves as a resolved value instruction
(\eg for the register resolve function).

\mypara{Fetch}
We now discuss the inference rules for memory operations, starting
with the fetch stage.
\begin{mathpar}
  \infer[simple-fetch]
  { \mem(\n) \in \{ \op, \load, \store, \later{\lfence} \} \\
    n' = \mathit{next}(\mem(\n)) \\
    i = \freshi \\
    \buf' = \buf[i \aug \mathit{transient}(\mem(\n))] }
  { \specstep{\fetch} (\reg, \mem, \n', \buf') }
\end{mathpar}
Given a $\fetch$ directive, rule \textsc{simple-fetch} extends the
reorder buffer $\buf$ with a new transient instruction (see
Table~\ref{tab:instructions}). Other than $\load$ and $\store$, the rule
also applies to $\op$ and $\lfence$ instructions.
The $\mathit{transient(\cdot)}$
function simply translates the physical instruction at $\mem(\n)$ to
its unresolved transient form.
It inserts the new, transient instruction at the first empty index in $\buf$,
and sets the current program point to the next instruction~$n'$.
Note that $\mathit{transient(\cdot)}$ annotates the transient $\load$
instruction with its program point.

\mypara{Load execution}
Next, we cover the rules for the $\load$ execute stage.
\begin{mathpar}
  \infer[load-execute-nodep]
  { \buf(i) = \iload{\rf}{\vec \rv}^n  \and
   \later{ \forall j < i : \buf(j) \neq \lfence} \\
    \regmap{\vec \rv} = \vec {v_\l} \and
    \addrf{\vec{v_\l}} = \a \\
    \l_a = \join \vec \l\\
    \forall j < i : \buf(j) \neq \istore{\_}{\a} \\
    \mem(\a) = v_{\l}  \\
    \buf' = \buf[i \aug \ival{\rf}{v _{\l}\{\bot, \a \}}^n] }
  { \specstep[\oread {\a_{\l_a}}]{\executei} (\reg, \mem, \n, \buf') }
\end{mathpar}\begin{mathpar}
  \infer[load-execute-forward]
  { \buf(i) = \iload{\rf}{\vec \rv}^n \and
   \later{ \forall j < i : \buf(j) \neq \lfence} \\
    \regmap{\vec \rv} = \vec{v _ \l} \and
    \addrf{\vec{v _ \l}} = \a \\
    \l_a = \join \vec \l\\
     max(j) < i : \buf(j) = \istore{\_}{\a} \wedge \buf(j)=\istore{v_\l, \a } \\
    \buf' = \buf[i \aug \ival{\rf}{v_{\l} \{j, \a\}} ^n] }
  { \specstep[\ofwd \a_{\l_a}]{\executei} (\reg, \mem, \n, \buf') }
\end{mathpar}
Given an $\executex$ directive for buffer index~$i$, under the
condition that $i$ points to an unresolved $\load$,
rule~\textsc{load-execute-nodep} applies if there are no prior
$\store$ instructions in $\buf$ that have a resolved, matching
address.
The rule first resolves the operand list~$\vec{\rv}$
into a list of values~$\vec{v_\l}$, and then uses $\vec{v_\l}$ to calculate the target
address $\a$. It then retrieves the current value~$v_\l$ at address $a$ from
memory, and finally adds to the buffer a resolved value instruction assigning
$v_\l$ to the target register $r$.
We annotate the value instruction with the address $a$ and
$\bot$, signifying that the value comes from memory.
Finally, the rule produces the observation
$\oread{\a_{\l_a}}$, which renders the memory read at address $\a$
with label $\l_a$ visible to an attacker.

Rule \textsc{load-execute-forward} applies if the most recent
$\store$ instruction in $\buf$ with a resolved, matching address has a
resolved data value.
Instead of accessing memory, the rule forwards the value from the
$\store$ instruction, annotating the new value instruction with the
calculated address $a$ and the index $j$ of the originating $\store$ instruction.
The rule produces a
$\ofwdx$ observation with the labeled address $a_{\l_a}$.
This observation captures that the attacker can determine (\eg by
observing the \emph{absence} of memory access using a cache timing attack)
that a
forwarded value from address $a$ was found in the buffer instead of loaded from memory.

Importantly, neither of the rules has to wait for prior $\store$s to
be resolved and can proceed speculatively. This can lead to
memory hazards when a more recent store to the $\load$'s address has not been
resolved yet; we show how to deal with hazards in the rules for the $\store$ instruction.

\mypara{Store execution}
We show the rules for $\store$s below.
\begin{mathpar}
  \infer[store-execute-value]
  { \buf(i) = \istore{\rv}{\vec \rv} \\
    \later{\forall j < i : \buf(j) \neq \lfence} \\
    \regmap{\rv} = v_\l  \\
    \buf' = \buf[i \aug \istore{v_\l}{\vec \rv}] }
  { \specstep{\executeivalue} (\reg, \mem, \n, \buf') }
\end{mathpar}\begin{mathpar}
  \infer[{store-execute-addr-ok}]
  { \buf(i) = \istore{\rv}{\vec \rv} \\
    \later{\forall j < i : \buf(j) \neq \lfence} \\
    \regmap{\vec \rv} = \vec {v_\l} \\
    \addrf{\vec{v_\l}} = \a \and   \l_a = \join \vec \l \\
    \forall k > i : \buf(k) = (r= \ldots \{j_k, \a_k\}) :  \\~
    ({\a}_k = \a \implies j_k \geq i)
      \later{{}\land~ (j_k = i \implies {\a}_k = \a)} \\\\
    \buf' = \buf[i \aug \istore{\rv}{\a_{\l_\a}}] }
  { \specstep[\ofwd \a_{\l_a}]{\execute{i:\,\textsf{addr}}} (\reg, \mem, \n, \buf') }
\end{mathpar}\begin{mathpar}
  \infer[store-execute-addr-hazard]
  { \buf(i) = \istore{\rv}{\vec \rv} \\
    \later{\forall j < i : \buf(j) \neq \lfence }\\
    \regmap{\vec \rv} = \vec{v_\l} \\
    \addrf{\vec{v_\l}} = \a\and   \l_a = \join \vec \l\\\\
     min(k) > i : \buf(k) = (r= \ldots \{j_k, \a_k \})^{n_k} : \\~
    (\a_k = \a \land j_k < i)
      \later{{}\lor~ (j_k = i \land \a_k \neq \a)} \\\\
    \buf' = \buf[j : j < k][i \aug \istore{\rv}{\a_{\l_\a}}] }
  { \specstep[\omiss,\ofwd \a_{\l_a}]{\execute{i:\,\textsf{addr}}} (\reg, \mem, n_k, \buf') }
\end{mathpar}
The execution of $\store$ is split into two steps: value resolution,
represented by the directive $\executeivalue$, and address resolution,
represented by the directive $\executeiaddr$;
a schedule may have either step first.
Either step may be skipped if data or address are already in
immediate form.

Rule \textsc{{store-execute-addr-ok}} applies if no misprediction
has been detected, \ie if no $\load$ instruction forwarded data from an outdated $\store$.
We check this by requiring that all value instructions \emph{after} the current index (indices $k > i$)
with an address $\a$ matching the current $\store$ must be using a value forwarded from a $\store$
\emph{at least as recent} as this one ($a_k = a \implies j_k \geq i$).
We define~$\bot < n$ for any index $n$---that is, if a future $\load$ matches the address of
the current $\store$ but loaded its value from memory, we consider this a hazard.

If there is indeed a hazard, \ie if there was a resolved
load with an outdated value, the rule \textsc{store-execute-addr-hazard} picks
the \emph{earliest} such instruction (index~$k$) and restarts execution by resetting
the instruction pointer to the program point $n_k$ of this instruction. It then discards all
transient instructions at indices at least $k$ from the
reorder buffer.
As in the case of misspeculation, the rule issues a $\omiss$ observation.

\mypara{Retire}
Resolved $\load$s are retired using the following rule.
\begin{mathpar}
  \infer[value-retire]
  { \retirable{i} \\
    \buf(i) = \ival{\rf}{v_\l} \\
    \reg' = \reg[r \aug v_\l] \\
    \buf' = \rmi{\buf} }
  { \specstep{\retire} (\reg', \mem, \n, \buf') }
\end{mathpar}
This is the same retire rule used for simple value instructions (\eg resolved $\op$ instructions).
The rule updates the register map $\reg$ with the new value, and removes the
instruction from the reorder buffer.

Stores are retired using the rule below.
\begin{mathpar}
  \infer[store-retire]
  { \retirable{i} \\
    \buf(i) = \istore{v_{\l} }{\a_{\l_a} } \\
    \mem' = \mem[\a \aug v _{\l} ] \\
    \buf' = \rmi{\buf} }
  { \specstep[\owrite \a_{\l_a} ]{\retire} (\reg, \mem', \n, \buf') }
\end{mathpar}
A fully resolved $\store$ instruction retires similarly to a value
instruction.
However, instead of updating the register map $\reg$,
rule~\textsc{store-retire} updates the memory $\mem$.
Since an attacker can observe memory writes, the rule produces the
observation $\owrite{\a_{\l_a}}$ with the labeled address of the $\store$.

\begin{figure}
  \small
  \centering

  \begin{tabular}{ll} 
  Registers & $\reg(r_a) = 40_\kpub$ \\
   Directives & D= $ \execute{\idxnum 4};  \executeaddr{\idxnum 3} $   \\
    Leakage for D&  $\ofwd \mathtt{43}_\kpub ; \omiss,\ofwd \mathtt{43}_\kpub $
  \end{tabular}
  \\
  
    \vspace{0.5em}

\begin{tabular}{lll}
 starting $\buf$      &  $\buf$ after     $ \execute{\idxnum 4}$             & $\buf$ after   D                            \\ 
 \hline
 $\idxnum 2~\istore{12}{\mathtt{43}_\kpub} $ & $\idxnum 2 ~\istore{12}{\mathtt{43}_\kpub} $   & $\idxnum 2 ~ \istore{12}{\mathtt{43}_\kpub} $   \\
 $\idxnum 3 ~ \istore{20}{[\mathtt{3},r_a]} $     & $\idxnum 3 ~\istore{20}{[\mathtt{3},r_a]}$         & $\idxnum 3 ~ \istore{20}{\mathtt{43}_\kpub}$  \\
 $\idxnum 4 ~ \iload{r_c}{[\mathtt{43}]} $         & $ \idxnum 4~ \ival{r_c}{12\{\idxnum 2,\mathtt{43}\}}$    \\

  \end{tabular}

\caption{Store hazard caused by late execution of $\store$ addresses. The $\store$ address
  for $\idxnum 3$ is resolved
  too late, causing the later $\load$ instruction to forward from the
  wrong $\store$. When $\idxnum 3$'s address is resolved, the execution must be rolled back.
  In this example, $\addrf{\cdot}$ adds its arguments.}
\label{fig:store-hazard}
\end{figure}

\begin{figure}
  \small
  \centering

  \begin{tabular}{llcrl}

    \multicolumn{2}{c}{Registers}                      && \multicolumn{2}{c}{Reorder buffer} \\
    $~r $               & $ \reg(r) $                  && $i~$        & $ \buf(i) $ \\
    \cline{1-2} \cline {4-5}
    $r_a$               & $ 5_\kpub $                  && $\idxnum 1 $ & $ \icond{\texttt{>}}{(4, r_a)}{\pcnum 2}{\pcnum 2}{\pcnum 4} $ \\
    $r_b$               & $ x_\ksec $                  && $\idxnum 2 $ & $ \istore{r_b}{[\mathtt{40}, r_a]} $ \\
    \multicolumn{2}{c}{Memory}                         &&              & \dots \\
    $~a $               & $ \mem(a) $                  && $\idxnum 7 $ & $ \iload{r_c}{[\mathtt{45}]} $ \\
    \cline{1-2}
    $ \mathtt{40..43} $ & $ \textit{secretKey}_\ksec $ && $\idxnum 8 $ & $ \iload{r_c}{[\mathtt{48}, r_c]} $ \\
    $ \mathtt{44..47} $ & $ \textit{pubArrA}_\kpub $   && \\
    $ \mathtt{48..4B} $ & $ \textit{pubArrB}_\kpub $   && \\

  \end{tabular}
  \vspace{0.5em}\\
  \begin{tabular}{lll}

    Directive                    & Effect on $\buf$                                                   & Leakage \\\hline
    $ \executeaddr{\idxnum 2} $  & $ \idxnum 2 \mapsto \istore{r_b}{\mathtt{45}_\kpub} $              & $ \ofwd \mathtt{45}_\kpub $ \\
    $ \executevalue{\idxnum 2} $ & $ \idxnum 2 \mapsto \istore{x_\ksec}{\mathtt{45}_\kpub} $          & \\
    $ \execute{\idxnum 7} $      & $ \idxnum 7 \mapsto \ival{r_c}{x_\ksec\{\idxnum 2,\mathtt{45}\}} $ & $ \ofwd \mathtt{45}_\kpub $ \\
    $ \execute{\idxnum 8} $      & $ \idxnum 8 \mapsto \ival{r_c}{X\{\bot,a\}} $                      &  \color{darkred}{ $\oread a_\ksec $} \\

  \end{tabular}

  where $a = x_\ksec + \mathtt{48}$
  \caption{
    Example demonstrating a store-to-load Spectre v1.1 attack.
    A speculatively stored value is forwarded and then leaked
    using a subsequent load instruction.
  }
\label{fig:v1.1-attack}
\end{figure}

\begin{figure}
  \small
  \centering

  \begin{tabular}{llcrl}

    \multicolumn{2}{c}{Registers}                    && \multicolumn{2}{c}{Reorder buffer} \\
    $~r $               & $ \reg(r) $                && $i~$         & $ \buf(i) $ \\
    \cline{1-2} \cline {4-5}
    $r_a$               & $ 40_\kpub $               && $\idxnum 2 $ & $ \istore{0}{[\mathtt{3},r_a]} $ \\
    \multicolumn{2}{c}{Memory}                       && $\idxnum 3 $ & $ \iload{r_c}{[\mathtt{43}]} $ \\
    $~a $               & $ \mem(a) $                && $\idxnum 4 $ & $ \iload{r_c}{[\mathtt{44}, r_c]} $ \\
    \cline{1-2}
    $ \mathtt{40..43} $ & $ \textit{secretKey}_\ksec $ && & \\
    $ \mathtt{44..47} $ & $ \textit{pubArrA}_\kpub $ && & \\

  \end{tabular}
  \vspace{0.5em}\\
  \begin{tabular}{lll}

    Directive                   & Effect on $\buf$                                                            & Leakage \\\hline
    $ \execute{\idxnum 3} $     & $ \idxnum 3 \mapsto \ival{r_c}{\textit{secretKey}[3]\{\bot,\mathtt{43}\}} $ & $ \oread \mathtt{43}_\kpub $ \\
    $ \execute{\idxnum 4} $     & $ \idxnum 4 \mapsto \ival{r_c}{X\{\bot,a\}} $                               & \color{darkred}{$ \oread a_\ksec $ } \\
    \multirow{2}{*}{$ \executeaddr{\idxnum 2} $}
            & $ \{\idxnum 3,\idxnum 4\} \notin \buf $
            & $ \omiss, $ \\
    & $ \idxnum 2 \mapsto \istore{0}{\mathtt{43}_\kpub} $
            & $ ~\ofwd \mathtt{43}_\kpub  $\\

  \end{tabular}

  where $a = \textit{secretKey}[3]_\ksec + \mathtt{44}$
\caption{Example demonstrating a v4 Spectre attack. The $\store$ is executed
  too late, causing later $\load$ instructions to use outdated values. }
\label{fig:v4-attack}
\end{figure}

\mypara{Example}
\Cref{fig:store-hazard} gives an example of store-to-load forwarding. In the
starting configuration, the $\store$ at index $\idxnum 2$ is fully resolved,
while the $\store$ at index $\idxnum 3$ has an unresolved address.
The first directive executes the $\load$ at $\idxnum 4$. This
$\load$ accesses address $\mathtt{43}$, which matches the $\store$ at
index $\idxnum 2$. Since this is the most recent such $\store$ and has a
resolved value, the $\load$ gets the value $12$ from this $\store$.
The following directive resolves the address of the $\store$ at index $\idxnum 3$.
This $\store$ also matches address $\mathtt{43}$. As this $\store$ is more recent than
$\store$ $\idxnum 2$, this directive triggers a hazard for the $\load$ at
$\idxnum 4$, leading to the rollback of the $\load$ from the reorder buffer.

\mypara{Capturing Spectre}
We now have enough machinery to capture several variants of Spectre attacks.

We discussed how our semantics model Spectre v1 in \Cref{sec:motivating}
(\Cref{fig:v1-attack}).
\Cref{fig:v1.1-attack} shows a simple disclosure gadget using forwarding from
an out-of-bounds write. In this example, a secret value $x_\ksec$ is supposed
to be written to $\textit{secretKey}$ at an index $r_a$ as long as $r_a$ is
within bounds. However, due to branch misprediction, the $\store$ instruction
is executed despite $r_a$ being too large. The $\load$ instruction at index
$\idxnum 7$, normally benign, now aliases with the $\store$ at index $\idxnum
2$, and receives the secret $x_\ksec$ instead of a public value from
$\textit{pubArrA}$. This value is then used as the address of another $\load$
instruction, causing $x_\ksec$ to leak.

\Cref{fig:v4-attack} shows a Spectre v4 vulnerability caused
when a $\store$ \emph{fails} to forward to a future $\load$. In this example,
the $\load$ at index $\idxnum 3$ executes before the $\store$ at $\idxnum 2$
calculates its address. As a result, this execution loads the
outdated secret value at address $\mathtt{43}$ and leaks it, instead of using
the public zeroed-out value that would be written.

\subsection{Aliasing Prediction}
\label{sec:alias}
We extend the memory semantics from the previous section to model aliasing
prediction by introducing a new transient instruction $\iload{\rf}{\vec \rv,
( v_\l, j)}^n$. This instruction represents a \emph{partially resolved} load with
speculatively forwarded data.
As before, $\rf$ is the target register, $\vec \rv$ is the list of
arguments for address calculation, and $n$ is the program point of the physical $\load$ instruction.  The new parameters are $v_\l$, the
forwarded data, and $j$, the index of the originating $\store$.

\mypara{Forwarding via prediction}
\begin{mathpar}
  \infer[load-execute-forwarded-guessed]
  { \buf(i) = \iload{\rf}{\vec \rv}^n \\
    j < i \\
    \later{\forall k < i : \buf(k) \neq \lfence} \\
    \buf(j) = \istore{v_{\l}}{\vec \rv_j} \\
    \buf' = \buf[i \aug \iload{\rf}{\vec \rv, ( v_{\l}, j)} ^n] }
  { \specstep{\executeig{j}} (\reg, \mem, \n, \buf') }
\end{mathpar}
Rule \textsc{load-execute-forwarded-guessed} implements forwarding in
the presence of unresolved target addresses.
Instead of forwarding the value from a $\store$ with a matching
address, as in \Cref{sec:forwarding}, the attacker can now freely
choose to forward from \emph{any} $\store$ with a resolved value---even if its
target address is not known yet.
Given a choice of which $\store$ $j$ to forward from---supplied via directive---the
rule updates the reorder buffer with the new partially resolved
$\load$ and records both the forwarded value $v_l$ and the buffer index $j$ of
the $\store$ instruction.

\mypara{Register resolve function}
We extend the register resolve function
$\regmapname$ to allow using values from partially resolved $\load$s.
In particular, whenever the register resolve function computes the latest
resolved assignment to some register~$r$, it now considers not only
fully resolved value instructions, but also our
new partially resolved $\load$: whenever the latest assignment in the
buffer is a partially resolved $\load$, the register resolve function
returns its value.

We now discuss the execution rules, where partially resolved $\load$s
may fully resolve against either the originating $\store$ or against memory.

\mypara{Resolving when originating store is in the reorder buffer}
\begin{mathpar}
  \infer[load-execute-addr-ok]
  { \buf(i) = \iload{\rf}{\vec \rv, ( v_{\l}, j )} ^n \\
    \regmap{\vec \rv} = \vec {v _\l} \\
    \addrf{\vec{v_\l}} = \a \\
    \l_a = \join \vec \l\\
    \buf(j) = \istore{v_{\l}}{\vec \rv_j} \land (\vec \rv_j = \a' \implies
     \a'=\a ) \\
    \forall k :(j<k<i) : \buf(k) \neq \istore{\_}{\a} \\
    \buf' = \buf[i \aug \ival{\rf}{v_{\l} \{j, \a\}} ^n] }
  { \specstep[\ofwd \a_{\l_a}]{\executei} (\reg, \mem, \n, \buf') }
\end{mathpar}\begin{mathpar}
  \infer[load-execute-addr-hazard]
  { \buf(i) = \iload{\rf}{\vec \rv, (v_{\l}, j )} ^{n'} \\
    \regmap{\vec \rv} = \vec {v_\l} \\
    \addrf{\vec{v_\l}} = \a \\
    \l_a = \join \vec \l\\
    ( \buf(j) = \istore{v_{\l}}{\a'} \land \a' \neq \a ) \mathrel{\lor} \\
      (\exists k: j<k<i \land \buf(k) = \istore{\_}{\a}) \\
    \buf' = \buf[j : j < i] }
  { \specstep[\omiss,\ofwd \a_{\l_a}]{\executei} (\reg, \mem, \n', \buf') }
\end{mathpar}
To resolve  $\iload{\rf}{\vec \rv, ( v_{\l}, j )} ^n$ when
  its originating $\store$ is still in $\buf$,
 we calculate the $\load$'s actual
target address $a$ and compare it against the target address of the
originating $\store$ at $\buf(j)$.
If the $\store$ is not followed by later $\store$s to $a$, and either
\textbf{(1)} the $\store$'s address is resolved and its address
is indeed~$a$, or \textbf{(2)} the $\store$'s address is still unresolved, we
update the reorder buffer with an annotated value instruction (rule \textsc{load-execute-addr-ok}).

If, however, either the originating $\store$ resolved to a \emph{different}
address (mispredicted aliasing) or a later $\store$ resolved to the same
address (hazard), we roll back our execution to just before the $\load$ (rule \textsc{load-execute-addr-hazard}).

We allow the $\load$ to execute even if the originating $\store$
has not yet resolved its address. When the $\store$ does finally resolve its
address, it must check that the addresses match and that the forwarding was correct.
The gray formulas in \textsc{store-execute-addr-ok} and \textsc{store-execute-addr-hazard} (\Cref{sec:forwarding})
perform these checks: For forwarding to be correct, all values forwarded from a $\store$ at
$\buf(i)$ must have a matching annotated address ($\forall k > i: j_k = i \implies a_k = a$).
Otherwise, if any value annotation has a mismatched address, then the instruction is rolled back
($j_k = i \land a_k \neq a$).

\mypara{Resolving when originating store is not in the buffer}
We~ must also consider the case where we have delayed resolving the $\load$ address
to the point where the originating $\store$ has already retired, and is no longer
available in $\buf$. If this is the case, and no other prior $\store$ instructions
have a matching address, then we must check the forwarded data against memory.

\begin{mathpar}
  \infer[load-execute-addr-mem-match]
  { \buf(i) = \iload{\rf}{\vec \rv, v_{\l}, j} ^n \\
    j \notin \buf \\
    \regmap{\vec \rv} = \vec {v_\l} \\
    \l_a = \sqcup \vec \l \\
    \addrf{\vec{v_\l}} = \a \\
    \forall k < i : \buf(k) \neq \istore{\_}{\a} \\
    \mem(\a) = v_{\l} \\
    \buf' = \buf[i \aug \ival{\rf}{v_{\l} \{\bot, \a\}} ^n] }
  { \specstep[\oread {\a_{\l_a}}]{\executei} (\reg, \mem, \n, \buf') }
\end{mathpar}\begin{mathpar}
  \infer[load-execute-addr-mem-hazard]
  { \buf(i) = \iload{\rf}{\vec \rv, v_{\l}, j} ^{n'} \\
    j \notin \buf \\
    \regmap{\vec \rv} = \vec {v_\l} \\
    \l_a  = \sqcup \vec \l \\
    \addrf{\vec v_\l} = \a \\
    \forall k < i : \buf(k) \neq \istore{\_}{\a} \\
    \mem(\a) = v'_{\l'} \\
    v'_{\l'} \neq v_\l \\
    \buf' = \buf[j : j < i] }
  { \specstep[\omiss,\oread {\a_{\l_a}}]{\executei} (\reg, \mem, \n', \buf') }
\end{mathpar}

If the originating $\store$ has retired, and no intervening $\store$s match the same address,
we must load the value from memory to ensure we were
originally forwarded the correct value.
If the value loaded from memory matches the value we were forwarded,
we update the reorder buffer with a resolved $\load$ annotated as if it had
been loaded from memory (rule \textsc{load-execute-addr-mem-match}).

If a $\store$ \emph{different} from the originating $\store$ overwrote
the originally forwarded value, the value loaded from memory may not
match the value we were originally forwarded. In this case we roll
back execution to just before the $\load$ (rule
\textsc{load-execute-addr-mem-hazard}).

We demonstrate these semantics in the 
attack shown in \Cref{fig:vnew-attack}.
An earlier draft of this paper~\cite{arxiv:v2} incorrectly claimed to have a
proof-of-concept exploit for this attack on real hardware.

\subsection{Speculation Barriers}
\label{sec:lfence}

We extend our semantics with a \emph{speculation
barrier} instruction, $\alfence \n$, that prevents further speculative
execution until all prior instructions have been retired.

\begin{mathpar}
  \infer[fence-retire]
  { \retirable{i} \\
    \buf(i) = \ilfence \\
    \buf' = \rmi{\buf} }
  { \specstep{\retire} (\reg, \mem, \n, \buf') }
\end{mathpar}

The $\lfence$ instruction uses \textsc{simple-fetch} as its $\fetch$ rule,
and its rule for $\retire$ only removes the instruction
from
the buffer. It does not have an
execute rule.
However, $\lfence$ instructions affect the
execution of all instructions in the reorder buffer that come
\emph{after} them.
In prior sections, execute rules have the highlighted condition $\forall j < i : \buf(j) \neq \lfence$.
This condition ensures that as long as a $\lfence$ instruction
remains in $\buf$, any instructions fetched after the
$\lfence$ cannot be executed.

We use $\lfence$ instructions to restrict
out-of-order execution in our semantics.
Notably, we can use it to prevent attacks of the forms shown in
\Cref{fig:v1-attack,fig:v1.1-attack,fig:v4-attack}.

\begin{figure}
  \small
  \centering

  \begin{tabular}{rlcrl}

    \multicolumn{2}{c}{Before executing $\idxnum 1$}                              && \multicolumn{2}{c}{After} \\
    $i~$         & $ \buf[i] $                                                    && $i~$         & $ \buf[i] $ \\
    \cline{1-2} \cline{4-5}
    $\idxnum 1 $ & $ \icond{\texttt{>}}{(4, r_a)}{\pcnum 2}{\pcnum 2}{\pcnum 5} $ && $\idxnum 1 $ & $ \ijmp{\pcnum 5} $ \\
    $\idxnum 2 $ & $ \ilfence $                                                   && \\
    $\idxnum 3 $ & $ \iload{r_b}{[\mathtt{40}, r_a]} $                            && \\
    $\idxnum 4 $ & $ \iload{r_c}{[\mathtt{44}, r_b]} $                            && \\

  \end{tabular}

  \caption{Example demonstrating fencing mitigation against Spectre v1 attacks. The $\lfence$ instruction
  prevents the $\load$ instructions from executing before the $\cond$.}
\label{fig:v1-fence}
\end{figure}

\mypara{Example}
The example in \Cref{fig:v1-fence} shows how placing a $\lfence$ instruction
just after the $\cond$ instruction prevents the Spectre v1 attack from
\Cref{fig:v1-attack}.
The $\lfence$ in this example prevents the $\load$ instructions at $\idxnum 2$
and $\idxnum 3$ from executing and forces the $\cond$
to be resolved first.
Evaluating the $\cond$ exposes the misprediction and causes the two $\load$s (as well as the $\lfence$)
to be rolled back.

\subsection{Indirect Jumps and Return Address Prediction}
\label{sec:indirect}
Finally, we briefly discuss two additional extensions to our semantics.
First, we extend our semantics with \emph{indirect jumps}. Rather than
specifying jump targets \emph{directly} as with the $\cond$ instruction
in \Cref{sec:cond}, indirect jumps compute the target from a list of
argument operands.
The indirect jump instruction has the form $\ajmpi{\vec \rv}$, where
$\vec \rv$ is the list of operands for calculating the jump target.
The transient form of $\jmpi$ is $\ijmpi{\vec \rv}{n_0}$, where $n_0$
is the predicted jump target.
To fetch a $\jmpi$ instruction, we use $\fetchg{\n'}$, where $\n'$ is the
speculated jump target. 
In all other respects, the rules for indirect jump instructions are similar
to the rules for conditional branches.

Second, we extend our semantics with \emph{call} and \emph{ret} instructions.
The call instruction has the form $\acall{n_f}{n_\mathit{ret}}$, where $n_f$ is
the callee program point and $n_\mathit{ret}$ is the program point to return
to. The return instruction is simply $\ret$.
Both the $\call$ and $\ret$ instructions have the simple transient forms $\icall$ and
$\iret$. However, when fetched, they are unpacked into multiple transient
instructions. Fetching a $\call$ produces the $\icall$ transient instruction as
well as transient instructions which will increment a stack pointer and store
the return program point to memory.
Fetching a $\ret$ produces a corresponding load, decrement, and
jump as well as the $\iret$ transient instruction.
Furthermore, the $\call$ and $\ret$ instructions respectively push and pop program points to
an additional configuration state representing the \emph{return stack buffer} (RSB).
The RSB is used to predict the new program point upon fetching a $\ret$.

In \Cref{sec:extended-semantics}, we present detailed rules for indirect jumps,
functions calls, and returns. We also show how both Spectre v2~\cite{spectre-v1} and
ret2spec~\cite{ret2spec} attacks can be expressed in our semantics, as well as
the \emph{retpoline} mitigation~\cite{turner} against Spectre v2 attacks.

\section{Detecting Violations}
\label{sec:tool}

We develop a tool \tool based on our semantics to check for SCT violations.
\tool first generates a set of schedules representing various \emph{worst-case}
attackers.
This set of schedules is far smaller than the set of all possible schedules for the
program, but is nonetheless sound:
if there is an SCT violation in any possible schedule, then there will be an
SCT violation in one of the worst-case schedules.
\tool then checks for secret leakage by symbolically executing the program
under each schedule.

\tool only exercises a subset of our semantics;
it does not detect SCT violations based on alias prediction, indirect jumps,
or return stack buffers (\Cref{sec:indirect,sec:alias}).
Doing so would require it to generate a prohibitively large number of
schedules.
Nevertheless, \tool still exposes attacks based on Spectre variants 1, 1.1, and
4.

We describe our schedule generation in \Cref{sec:alg},
and evaluate \tool on several crypto libraries in \Cref{sec:eval}.

\subsection{Schedule Generation}
\label{sec:alg}

Given a program, \tool generates a set of schedules representing various
worst-case attackers.
\tool's schedule generation is parametrized by a \emph{speculation bound},
which limits the size of the reorder buffer, and thus the depth of
speculation.

In general, \tool constructs worst-case schedules to maximize speculation.
These schedules eagerly fetch instructions until the reorder buffer is
full, \ie the size of the reorder buffer equals the speculation bound.
Once the reorder buffer is full, the schedules only retire instructions as
necessary to fetch new ones.

When conditional branches are to be fetched, \tool constructs schedules
containing both possible outcomes:
one where the branch is guessed true (\fetchg \ktrue) and one where the
branch is guessed false (\fetchg \kfalse).
For the mispredicted outcome, \tool's schedules execute the branch as late as
possible (\ie it is the oldest instruction in the reorder buffer and the
reorder buffer is full), which delays the rollback of mispredicted paths.

To account for the load-store forwarding hazards described in
\Cref{sec:forwarding}, \tool constructs schedules containing all possible
forwarding outcomes.
For every load instruction $l$ in the program, \tool finds all prior stores
$s_i$ within the speculation bound that would resolve to the same address.
Then, for each such store, \tool constructs a schedule that would cause that
store to forward its data to $l$.
That is, \tool constructs separate schedules [$\executeaddr{s_1};
\execute{l}$], [$\executeaddr{s_2}; \execute{l}$], and so on.
Additionally, \tool constructs a schedule where none of the prior stores
$s_i$ have resolved addresses, forcing the load instruction to read from
memory.

For all instructions other than conditional branches and memory
operations, \tool only constructs schedules where these instructions are
executed eagerly and in order.
Reordering of these instructions is uninteresting:
either the instructions naturally commute, or data dependencies prevent the
reordering (\ie the reordered schedule is invalid for the program).
This intuition matches with the property that any out-of-order
execution of a given program has the same final result regardless of its
schedule.

We formalize the soundness of \tool's schedule construction in more
detail in \Cref{sec:tool-soundness}.

\subsection{Implementation and Evaluation}
\label{sec:eval}

We implement \tool on top of the \texttt{angr} binary-analysis
tool~\cite{shoshitaishvili2016state}.
\tool uses \texttt{angr} to symbolically execute a given program
according to each of its worst-case schedules, flagging any resulting
secret leakage.

To sanity check \tool, we create and analyze a set of Spectre v1 and v1.1
test cases, and ensure we flag their SCT violations.
Our test cases are based off the well-known Kocher Spectre v1 examples~\cite{kocher-tests}.
Since many of the Kocher examples exhibit violations even during sequential
execution, we create a new set of Spectre v1 test cases which only exhibit
violations when executed speculatively.
We also develop a similar set of test cases for Spectre v1.1 data attacks.

\tool necessarily inherits the limitations of \texttt{angr}'s symbolic
execution.
For instance, \texttt{angr} concretizes addresses for
memory operations instead of keeping them symbolic.
Furthermore, exploring every speculative
branch and potential store-forward within a given speculation bound leads to
an explosion in state space.
In our tests, we were able to support speculation bounds of up to 20
instructions.
We were able to increase this bound to 250 instructions when we disabled
checking for store-forwarding hazards.
Though these bounds do not capture the speculation depth of some modern
processors, \tool still correctly finds SCT violations in all our test cases,
as well as SCT violations in real-world crypto code.
We consider the design and
implementation of a more scalable tool future work.

\subsubsection{Evaluation Procedure}
To evaluate \tool on real-world crypto implementations, we use
the same case studies as FaCT~\cite{cauligi:2019:fact}, a
domain-specific language and compiler for constant-time crypto code.
We use FaCT's case studies for two reasons:
these implementations have been verified to be (sequentially) constant-time,
and their inputs have already been annotated by the FaCT authors with secrecy labels.\footnote{\url{https://github.com/PLSysSec/fact-eval}}

We analyzed both the FaCT-generated binaries and the corresponding C binaries for
the case studies. For each binary, we ran \tool without forwarding hazard
detection---only looking for Spectre v1 and v1.1 violations---and with a
speculation bound of 250 instructions. If \tool did not flag any violations, we
re-enabled forwarding hazard detection---looking for Spectre v4 violations---and
ran \tool with a reduced bound of 20 instructions. The reduced bound
ensured that the analysis was tractable.

\begin{table}[t]
  \definecolor{checkgreen}{rgb}{0,.65,0}
  \newcommand\no{\checkmark}
  \newcommand\yes[1]{$\mathghost$} %
  \newcommand\fwd[1]{\yes{}$^f$} %
  \caption[]{A \yes{}\ indicates \tool found an SCT violation.
  A \fwd{}\ indicates the violation was found only with forwarding hazard detection.}
  \label{tab:eval}
  \begin{tabular}{l|c|c}
    \toprule
      \textbf{Case Study}          & \textbf{C}    & \textbf{FaCT} \\\midrule
      curve25519-donna             & \no           & \no           \\
      libsodium \texttt{secretbox} & \yes{lc-msg } & \no           \\
      OpenSSL ssl3 record validate & \yes{emctx-r} & \fwd{mispop}  \\
      OpenSSL MEE-CBC              & \yes{sha1fin} & \fwd{misret}  \\
    \bottomrule
  \end{tabular}
\end{table}

\subsubsection{Detected Violations}
\Cref{tab:eval} shows our results.
\tool did not flag any SCT violations in the curve25519-donna implementations; this is not surprising,
as the curve25519-donna library is a straightforward implementation of crypto primitives.
\tool did, however, find SCT violations (without forwarding hazard detection)
in both the libsodium and OpenSSL codebases.
Specifically, \tool found violations in the C implementations of these
libraries, in code ancillary to the core crypto routines.
This aligns with our intuition that crypto primitives will not themselves be
vulnerable to Spectre attacks, but higher-level code that interfaces with
these primitives may still leak secrets.
Such higher-level code is not present in the corresponding FaCT
implementations, and \tool did not find any violations in the FaCT code with
these settings.
However, with forwarding hazard detection, \tool was able to find
vulnerabilities even in the FaCT versions of the OpenSSL implementations.
We describe two of the violations \tool flagged next.

\mypara{C libsodium secretbox}
The libsodium codebase compiles with stack protection~\cite{gccstackprotect} turned on by default.
This means that, for certain functions (\eg functions with stack allocated
\ccode{char} buffers), the compiler inserts code in the function epilogue to
check if the stack was ``smashed''.
If so, the program displays an error message and aborts.
As part of printing the error message, the program calls a function 
\ccode{__libc_message}, which does \ccode{printf}-style string formatting. 

\begin{figure}[t]
\begin{minted}{c}
for (int cnt = nlist - 1; cnt >= 0; --cnt) {
  iov[cnt].iov_base = (char *) list->str;
    // ...
  list = list->next;
}
\end{minted}
  \caption[]{Vulnerable snippet from \ccode{__libc__message()}.\footnotemark}
  \label{fig:libc-message}
\end{figure}
\footnotetext{Code snippet
  taken from \url{https://github.com/lattera/glibc/blob/895ef79e04a953cac1493863bcae29ad85657ee1/sysdeps/posix/libc_fatal.c}}
\Cref{fig:libc-message} shows a snippet from this function which traverses a linked list.
When running the C \ccode{secretbox} implementation speculatively, the
processor may misspeculate on the stack tampering check and jump into the
error handling code, eventually calling \ccode{__libc_message}.
Again due to misspeculation, the processor may incorrectly proceed through
the loop extra times, traversing non-existent links, eventually causing
secret data to be stored into \ccode{list} instead of a valid address (line
4).
On the following iteration of the loop, dereferencing \ccode{list} (line 2)
causes a secret-dependent memory access.

\mypara{FaCT OpenSSL MEE}
\begin{figure}[t]
\begin{minted}{fact_lexer.py:FaCTLexer -x}
  aesni_cbc_encrypt(/* ... */);
  // (len _out) is in %r14
  secret mut uint32 pad = _out[len _out - 1];
  public uint32 maxpad = tmppad > 255 ? 255 : tmppad;
  if (pad > maxpad) {
    pad = maxpad;
    ret = 0; // overwrites %r14
  }
  // ...
  _sha1_update(/* ... */); // can return to line 3
\end{minted}
  \caption[]{Vulnerable snippet from the FaCT OpenSSL MEE implementation.\footnotemark}
  \label{fig:mee-misreturn}
\end{figure}
\footnotetext{Code snippet
  taken from \url{https://github.com/PLSysSec/fact-eval/blob/888bc6c6898a06cef54170ea273de91868ea621e/openssl-mee/20170717_latest.fact}}
\newcommand{\thereg}{\texttt{\%r14}\xspace}
In \Cref{fig:mee-misreturn},
we show the code from the FaCT port of OpenSSL's authenticated encryption implementation.
The FaCT compiler transforms the branch at lines 5-7
into straight-line constant-time code, since the variable \fcode{pad} is considered \fcode{secret}.

Initially, register \thereg holds the length of the array \fcode{_out}.
The processor leaks this value due to the array access on line 3; this is
not a security violation, as the length is public.
On line 7, the value of \thereg is overwritten with \texttt{0} if
\fcode{pad > maxpad}, or \texttt{1} (the initial value of \fcode{ret})
otherwise.
Afterwards, the processor calls \fcode{_sha1_update}.

To return from \fcode{_sha1_update}, the processor must first load the return address
from memory.
When forwarding hazard detection is enabled, \tool allows this load to
speculatively receive data from stores \emph{older} than the most recent
store to that address (see \Cref{sec:forwarding}).
Specifically, it may receive the prior value that was stored at that
location: the return address for the call to \fcode{aesni_cbc_encrypt}.

After the speculative return, the processor executes line~3 a second time. 
This time, \thereg does not hold
the public value \fcode{len _out}; it instead holds the value of \fcode{ret},
which was derived from the secret condition \fcode{pad > maxpad}.
The processor thus accesses either \fcode{_out[0]} or \fcode{_out[-1]}, leaking
information about the secret value of \fcode{pad} via cache state.

\section{Related Work}
\label{sec:related}
Prior work on modeling speculative or out-of-order execution is concerned
with correctness rather than
security~\cite{lahiri2002modeling,alglave2009semantics}.
We instead focus on security and model side-channel leakage
explicitly.
Moreover, we abstract away the specifics of microarchitectural
features, considering them to be adversarially controlled.

Disselkoen et al.~\cite{disselkoen19code} explore speculation and
out-of-order effects through a relaxed memory model. Their semantics
sits at a higher level, and is orthogonal to our approach. They do not
define a semantic notion of security that prevents Spectre-like
attacks, and do not provide support for verification.

Mcilroy et al.~\cite{mcilroy19spectre} reason about micro-architectural attacks
using a multi-stage pipeline semantics (though they do not define a formal
security property).
Their semantics models branch predictor and cache state explicitly.
However, they do not model the effects of speculative barriers, nor other
microarchitecture features such as store-forwarding.
Thus, their semantics can only capture Spectre v1 attacks.

Both Guarnieri et al.~\cite{guarnieri18spectector} and Cheang et
al.~\cite{cheang19formal} define speculative semantics that are supported by tools.
Their semantics handle speculation through branch prediction---where the
predictor is left abstract---but do not capture more general out-of-order
execution nor other types of speculation.
These works also propose new semantic notions of security (different from SCT);
both essentially require that the speculative execution of a program not leak
more than its sequential execution.
If a program is sequentially constant-time, this additional security
property is equivalent to our notion of speculative constant-time.
Though our property is stronger, it is also simpler to verify: we can directly
check SCT without first checking if a program is sequentially constant-time.
And since we focus on cryptographic code, we directly require the stronger SCT
property.

Balliu et al.~\cite{balliu2019inspectre} define a semantics in a style similar
to ours. Their semantics captures various Spectre attacks, including an attack
similar to our alias prediction example (\Cref{fig:vnew-attack}), and a new
attack based on their memory ordering semantics, which our semantics cannot
capture.

Finally, several tools detect Spectre vulnerabilities, but do not
present semantics.
The oo7 static analysis tool~\cite{wang18oo7}, for example, uses taint tracking
to find Spectre attacks and automatically insert mitigations for several
variants.
Wu and Wang~\cite{wu19abstract}, on the other hand, perform cache analysis of
LLVM programs under speculative execution, capturing Spectre v1
attacks.

\section{Conclusion}
We introduced a semantics for reasoning about side-channels under adversarially
controlled out-of-order and speculative execution.  Our semantics capture
existing transient execution attacks---namely Spectre---but can be extended to
future hardware predictors and potential attacks.  We also defined a new notion
of constant-time code under speculation---speculative constant-time (SCT)---and
implemented a prototype tool to check if code is SCT.  Our prototype, \tool,
discovered new vulnerabilities in real-world crypto libraries.

There are several directions for future work. Our immediate plan is to use our
semantics to prove the effectiveness of existing countermeasures (\eg
retpolines) and to justify new countermeasures.

\begin{acks}

We thank the anonymous PLDI and PLDI AEC reviewers and our shepherd James
Bornholt for their suggestions and insightful comments.
We thank David Kaplan from AMD for his detailed analysis of our
proof-of-concept exploit that we incorrectly thought to be abusing an aliasing
predictor.
We also thank Natalie Popescu for her aid in editing and formatting this paper.
This work was supported in part by gifts from Cisco and Fastly,  %
by the NSF under Grant Number CCF-1918573,  %
by ONR Grant N000141512750,  %
and by the CONIX Research Center, one of six centers in JUMP,  %
a Semiconductor Research Corporation (SRC) program sponsored by DARPA.  %

\end{acks}

\bibliographystyle{ACM-Reference-Format}
\bibliography{local}

\appendix

\section{Extended semantics}
\label{sec:extended-semantics}

\subsection{Indirect jumps}
\label{sec:jmpi-semantics}

\mypara{Semantics}
The semantics for $\jmpi$ are given below:
\begin{mathpar}
  \infer[jmpi-fetch]
  { \mem(\n) = \ajmpi{\vec \rv} \\
    i = \freshi \\
    \buf' = buf[i \aug \ijmpi{\vec \rv}{n'}] }
  { \specstep{\fetchg {n'}} (\reg, \mem, \n', \buf') }

  \infer[jmpi-execute-correct]
  { \buf(i) = \ijmpi{\vec \rv}{n_0} \\
    \forall j < i : \buf(j) \neq \lfence \\
    \regmap{\vec \rv} = \vec {v_\l}  \\
    \l = \join \vec \l \\
    \addrf{\vec{v_\l}} = n_0 \\
    \buf' = \buf[i \aug \ijmp{n_0}] }
  { \specstep[\ojump {n_0}_\l]{\executei} (\reg, \mem, \n, \buf') }

  \infer[jmpi-execute-incorrect]
  { \buf(i) = \ijmpi{\vec \rv}{n_0} \\
    \forall j < i : \buf[j] \neq \lfence \\
    \regmap{\vec \rv} = \vec {v_\l}  \\
    \l = \join \vec \l \\
    \addrf{\vec{v_\l}} = n' \neq n_0 \\
    \buf' = \buf[j : j < i][~ i \aug \ijmp{n'}] }
  { \specstep[\omiss,\ojump n'_{\l}]{\executei} (\reg, \mem, \n', \buf') }
\end{mathpar}

When fetching a $\jmpi$ instruction, the schedule guesses the jump
target $n'$.  The rule records the operands and the guessed program point 
in a new buffer entry.
In a real processors, the jump target guess is supplied by an indirect
branch predictor; as branch predictors can be arbitrarily influenced
by an adversary~\cite{spectre-ng}, we model the guess as an attacker
directive.

In the execute stage, we calculate the actual jump target and compare
it to the guess. If the actual target and the guess match, we update
the entry in the reorder buffer to the resolved jump instruction
$\ijmp n_0$.
If actual target and the guess do not match, we roll back the
execution by removing all buffer entries larger or equal to~$i$,
update the buffer with the resolved jump to the correct address, and
set the next instruction.

Like conditional branch instructions, indirect
jumps leak the calculated jump target.

\begin{figure}
  \small
  \centering

  \begin{tabular}{llcrl}

    \multicolumn{2}{c}{Registers}                      && \multicolumn{2}{c}{Program} \\
    $~r $               & $ \reg(r) $                  && $n $        & $ \mem(n) $ \\
    \cline{1-2} \cline {4-5}
    $r_a$               & $ 1_\kpub $                  && $ \pcnum 1   $ & $ \aload{r_c}{\mathtt{[48,r_a]}}{\pcnum 2} $ \\
    $r_b$               & $ 8_\kpub $                  && $ \pcnum 2   $ & $ \alfence{\pcnum 3} $ \\
    \multicolumn{2}{c}{Memory}                         && $ \pcnum 3   $ & $ \ajmpi{[\mathtt{12},r_b]} $ \\
    $~a $               & $ \mem(a) $                  &&                & \dots \\
    \cline{1-2}
    $ \mathtt{44..47} $ & $ \textit{array B}_\kpub $   && $ \pcnum{16} $ & $ \alfence{\pcnum{17}} $ \\
    $ \mathtt{48..4B} $ & $ \textit{array Key}_\ksec $ && $ \pcnum{17} $ & $ \aload{r_d}{\mathtt{[44,r_c]}}{\pcnum{18}} $ \\
  \end{tabular}
  \vspace{0.5em}\\
  \begin{tabular}{lll}

    Directive             & Effect on $\buf$                                             & Leakage \\\hline
    $\fetch$              & $ \idxnum 1 \mapsto \iloadx{r_c}{\mathtt{48} + r_a} $        & \\
    $\fetch$              & $ \idxnum 2 \mapsto \ilfence $                               & \\
    $\execute{\idxnum 1}$ & $ \idxnum 1 \mapsto \ivalx{r_c}{\textit{Key}[1]_\ksec} $     & $ \oread \mathtt{49}_\kpub $ \\
    $\fetchg{\pcnum{17}}$ & $ \idxnum 3 \mapsto \ijmpi{[\mathtt{12}, r_b]}{\pcnum{17}} $ & \\
    $\fetch$              & $ \idxnum 4 \mapsto \iloadx{r_d}{[\mathtt{44},r_c]} $        & \\
    $\retire$             & $ \idxnum 1 \notin \buf $                                    & \\
    $\retire$             & $ \idxnum 2 \notin \buf $                                    & \\
    $\execute{\idxnum 4}$ & $ \idxnum 4 \mapsto \ivalx{r_d}{X} $                         & $ \oread a_\ksec $ \\

  \end{tabular}

  where $a = \textit{Key}[1]_\ksec + \mathtt{40}$
  \caption{Example demonstrating Spectre v2 attack from a mistrained indirect branch predictor.
  Speculation barriers are not a useful defense against this style of attack.}
  \label{fig:v2-attack}
\end{figure}

\mypara{Examples}
The example in \Cref{fig:v2-attack} shows how a mistrained indirect branch predictor
can lead to disclosure vulnerabilities. After loading a secret value into $r_c$
at program point $\pcnum 1$, the program makes an indirect jump. An adversary
can mistrain the predictor to send execution to $\pcnum{17}$ instead of the intended
branch target, where the secret value in $r_c$ is immediately leaked.
Because indirect jumps can have arbitrary branch target locations, $\lfence$ instructions
do not prevent these kinds of attacks; an adversary can simply retarget the indirect
jump to the instruction after the $\lfence$, as is seen in this example.

\subsection{Return address prediction}
\label{sec:callret}
Next, we discuss how our semantics models function calls.

\mypara{Instructions}
We introduce the following two physical instructions:
$\acall{n_f}{n_\mathit{ret}}$, where $n_f$ is the target program point of the call and
$n_\mathit{ret} $ is the return program point; and the return instruction $\ret$.
Their transient forms are simply $\icall$ and $\iret$.

\mypara{Call stack}
To track control flow in the presence of function calls, our
semantics explicitly maintains a call stack in memory.
For this, we use a dedicated register $\rsp$ which points to the top
of the call stack, and which we call the \emph{stack pointer
register}.

On fetching a $\call$ instruction, we update $\rsp$ to point to the address
of the next element of the stack using an abstract operation $\mathop{succ}$.
It then saves the return address to the newly computed address.
On returning from a function call, our semantics transfers control to
the return address at $\rsp$, and then updates $\rsp$ to point to the
address of the previous element using a function~$\mathop{pred}$.
This step makes use of a temporary register
$\rret$.

Using abstract operations $\mathop{succ}$ and $\mathop{pred}$ rather
than committing to a concrete implementation allows
our semantics to capture different stack designs.
For example, on a 32-bit x86 processor with a downward-growing
stack, $\op(\mathop{succ}, \rsp)$ would be implemented as $\rsp - 4$,
while $\op(\mathop{pred}, \rsp)$ would be implemented as $\rsp + 4$; on
an upward growing system, the reverse would be true.

Note that the stack register $\rsp$ is not protected from illegal
access and can be updated freely.

\begin{figure}
  \small
  \centering
   Program 
    \begin{tabular}{llll}      
      $n $ &     $\pcnum 1$ &     $\pcnum 2 $ &    $\pcnum 3 $  \\
      $ \mem(n)$ &  $\acall{\pcnum{3}}{\pcnum{2}}$ & $\aret$ & $
                                                               \aret$ \\
  \end{tabular}

  \vspace{0.5em}
  \begin{tabular}{llll}

    Directive                             & $\n$                                      & $\buf$                                                   & $\rets$ \\\hline
    $ \fetch $                            & $\pcnum 1 \to \pcnum 3$                   & $ \idxnum 1 \mapsto \icall $                             & $\idxnum 1 \mapsto \push\ \pcnum 2$\\
                                          &                                           & $ \idxnum 2 \mapsto \iopx{\mathop{succ}}{\rsp}{\rsp} $   & \\
                                          &                                           & $ \idxnum 3 \mapsto \istore{\mathtt{2}}{[\rsp]} $        & \\
    $ \fetch $                            & $\pcnum 3 \to \pcnum 2$                   & $ \idxnum 4 \mapsto \iret $                              & $\idxnum 4 \mapsto \pop$\\
                                          &                                           & $ \idxnum 5 \mapsto \iloadx{\rret}{[\rsp]} $             & \\
                                          &                                           & $ \idxnum 6 \mapsto \iopx{\mathop{pred}}{\rsp}{\rsp} $   & \\
                                          &                                           & $ \idxnum 7 \mapsto \ijmpi{[\rret]}{\pcnum 2} $          & \\
    $\fetchg{\color{darkred}{\pcnum{n}}}$ & $\pcnum 2 \to \color{darkred}{\pcnum{n}}$ & $ \idxnum 8 \mapsto \iret $                              & $\idxnum 8 \mapsto \pop$\\
                                          &                                           & $ \idxnum 9 \mapsto \iloadx{\rret}{[\rsp]} $             & \\
                                          &                                           & $ \idxnum{10} \mapsto \iopx{\mathop{pred}}{\rsp}{\rsp} $ & \\
                                          &                                           & $ \idxnum{11} \mapsto \ijmpi{[\rret]}{{\color{darkred}{\pcnum n}}} $        & \\
  \end{tabular}

  \caption{Example demonstrating a ret2spec-style attack~\cite{ret2spec}.
  The attacker is able to send (speculative) execution to an arbitrary program point, shown in red.}
\label{fig:rsb-attack}
\end{figure}

\mypara{Return stack buffer}
For performance, modern processors speculatively predict return addresses.
To model this, we extend configurations with a new piece of state
called the \emph{return stack buffer} (RSB), written as $\rets$.
The return stack buffer contains the expected return address at any
execution point.
Its implementation is simple: for a $\call$ instruction, the semantics
pushes the return address to the RSB, while for a $\ret$ instruction, the
semantics pops the address at the top of the RSB.
Similar to the reorder buffer, we address the RSB through indices and
roll it back on misspeculation or memory hazards.

We model return prediction directly through the return stack
buffer rather than relying on attacker directives, as most processors
follow this simple strategy, and the predictions therefore cannot be
influenced by an attacker.

We now present the step rules for our semantics.

\mypara{Calling}
\renewcommand\specstep[2][]{(\reg, \mem, \n, \buf, \rets) \xhookrightarrow[#2]{#1} }
\begin{mathpar}
  \infer[call-direct-fetch]
  { \mem(\n) = \acall{n_f}{n_\mathit{ret}} \\
    i = \freshi \\
    \buf_1= \buf[ i \aug \icall][i+1 \aug \iop{\mathop{succ}}{\rsp}{\rsp}] \\
   \buf'= \buf_1 [i + 2 \aug \istore{n_\mathit{ret}}{[\rsp]}] \\ %
     \rets' = \rets[i \aug \push\ n_\mathit{ret}] \and \n'= \n_f }
  { \specstep{\fetch} (\reg, \mem, \n', \buf', \rets') }
\end{mathpar}\begin{mathpar}
  \infer[call-retire]
  { \retirable{i} \\
    \buf(i) = \icall \\
    \buf(i+1) = \ival{\rsp}{\vsp} \\
    \buf(i+2) = \istore{n_\mathit{ret} }{\a_{\l_a} } \\
    \reg' = \reg[\rsp \aug \vsp] \\
    \mem' = \mem[\a \aug n_\mathit{ret} ] \\
    \buf' = \buf[j : j > i + 2] }
 { \specstep[\owrite \a_{\l_a} ]{\retire} (\reg', \mem', \n, \buf', \rets) }
 \end{mathpar}

On fetching a $\call$ instruction, we add three transient instructions to the reorder buffer
to model pushing the return address to the in-memory stack.
The first transient instruction, $\icall$, simply serves as an indication that the following
two instructions come from fetching a $\call$ instruction.
The remaining two instructions advance $\rsp$ to point to a new stack entry, then store the return
address $n_\mathit{ret}$ in the new entry. Neither of these transient instructions are fully
resolved---they will need to be executed in later steps.
We next add a new entry to the RSB, signifying a push of the return address $n_\mathit{ret}$ to the RSB.
Finally, we set our program point to the target of the call $n_f$.

When retiring a $\call$, all three instructions generated during the fetch are retired together.
The register file is updated with the new value of $\rsp$, and the return address is written to physical memory,
producing the corresponding leakage.

The semantics for direct calls can be extended to cover indirect calls in a straightforward manner by
imitating the semantics for indirect jumps. We omit them for brevity.

\mypara{Evaluating the RSB}
We define a function $\sttop(\rets)$ that retrieves the value at the
top of the RSB stack.
For this, we let $\ll \rets\rr$ be a function that transforms the RSB
stack~$\rets$ into a stack in the form of a partial map ($\mathit{st}
: \mathcal{N} \rightharpoonup \V$) from the natural numbers to program
points, as follows: the function~$\ll \cdot \rr$ applies the commands
for each value in the domain of~$\rets$, in the order of the indices. For
a $\push\ n$ it adds~$n$ to the lowest empty index of $\mathit{st}$.
For $\pop$, it and removes the value with the highest index
in~$\mathit{st}$, if it exists. We then define $\sttop(\rets)$ as
$\mathit{st}(\textsf{MAX}(\mathit{st}))$, where $\mathit{st}=\ll
\rets\rr$, and $\bot$, if the domain of~$\mathit{st}$ is empty.
For example, if $\rets$ is given as $\emptyset [1 \aug\ \push \; \pcnum 4][2
  \aug\ \push \; \pcnum  5] [3 \aug pop]$, then
$\ll \rets \rr= \emptyset [1 \aug  \pcnum 4]$, and
$\sttop(\rets)=\pcnum 4$.

\mypara{Returning}
\begin{mathpar}
  \renewcommand\specstep[2][]{(\reg, \mem, \n, \buf, \rets) \xhookrightarrow[#2]{#1} }
  \infer[ret-fetch-rsb]
  { \mem(\n) = \aret \\
    \sttop(\rets) = n' \\
    i = \freshi \\
    \buf_1 = \buf[i \aug \iret] \\
    \buf_2 = \buf_1[i + 1\aug \iload{\rret}{[\rsp]}] \\
    \buf_3 = \buf_2[i + 2 \aug \iop{\mathop{pred}}{\rsp}{\rsp}] \\
    \buf_4 = \buf_3[i + 3 \aug \ijmpi{[\rret]}{n'}] \\
    \rets' = \rets[i \aug \pop] }
  { \specstep{\fetch} (\reg, \mem, \n', \buf_4, \rets') }
\end{mathpar}\begin{mathpar}
  \infer[ret-fetch-rsb-empty]
  { \mem(\n) = \aret \\
    \sttop(\rets) = \bot \\
    i = \freshi \\
    \buf_1 = \buf[i \aug \iret] \\
    \buf_2 = \buf_1[i + 1 \aug \iload{\rret}{[\rsp]}] \\
    \buf_3 = \buf_2[i + 2 \aug \iop{\mathop{pred}}{\rsp}{\rsp}] \\
    \buf_4 = \buf_3[i + 3 \aug \ijmpi{[\rret]}{\n'}] \\
    \rets' = \rets[i \aug \pop] }
  { \specstep{\fetchg{\n'}} (\reg, \mem, \n', \buf_4, \rets') }
\end{mathpar}\begin{mathpar}
  \infer[ret-retire]
  { \retirable{i} \\
    \buf(i) = \iret \\
    \buf(i+1) = \ival{\rret}{{v_1}_{\l_1}} \\
    \buf(i+2) = \ival{\rsp}{{v_2}_{\l_2}} \\
    \buf(i+3) = \ijmp {n'} \\
    \reg' = \reg[\rsp \aug {v_2}_{\l_2}] \\
    \buf' = \buf[j : j > i + 3] }
  { \specstep{\retire} (\reg', \mem, \n, \buf', \rets) }
\end{mathpar}

On a fetch of $\ret$, the next program point is set  to the predicted
return address, \ie the top value of the RSB, $\sttop(\rets)$.
Just as with $\call$, we add the transient $\iret$ instruction to the reorder buffer, followed by
the following (unresolved) instructions:
we load the value at address $\rsp$ into a temporary register $\rret$,
we ``pop'' $\rsp$ to point back to the previous stack entry, and then add an indirect
jump to the program point given by $\rret$.
Finally, we add a $\pop$ entry to the RSB. As with $\call$
instructions, the set of instructions generated by a $\ret$ fetch are
retired all at once.

When the RSB is empty, the attacker can supply a speculative return
address via the directive $\fetchg{n'}$. This is consistent with the
behavior of existing processors. %
In practice, there are several variants on what processors actually do when the RSB is empty~\cite{ret2spec}:
\begin{itemize}
  \item AMD processors refuse to speculate. This can be modeled by defining
    $\sttop(\rets)$ to be a failing predicate if it would result in $\bot$.
  \item Intel Skylake/Broadwell processors fall back to using their branch target predictor. This can be modeled by
    allowing arbitrary $n'$ for the $\fetchg{n'}$ directive for the \textsc{ret-fetch-rsb-empty} rule.
  \item ``Most'' Intel processors treat the RSB as a circular buffer, taking whichever value is produced when the
    RSB over- or underflows. This can be modeled by having $\sttop(\rets)$ always produce an according value, and
    never producing $\bot$.
\end{itemize}

\begin{figure}
  \small
  \centering

  \begin{tabular}{llcrl}

    \multicolumn{2}{c}{Registers}                      && \multicolumn{2}{c}{Program} \\
    $~r $               & $ \reg(r) $                  && $n $        & $ \mem(n) $ \\
    \cline{1-2} \cline {4-5}
    $r_b$               & $ 8_\kpub $                  && $ \pcnum 3   $ & $ \acall{\pcnum 5}{\pcnum 4} $ \\
    $\rsp$              & $ \mathtt{7C}_\kpub $        && $ \pcnum 4   $ & $ \alfence{\pcnum 4} $ \\
                                                      &&& $ \pcnum 5   $ & $ \aopx{\addr}{r_d}{[\mathtt{12}, r_b]}{\pcnum 6} $ \\
                                                      &&& $ \pcnum 6   $ & $ \astore{r_d}{[\rsp]}{\pcnum 7} $ \\
                                                      &&& $ \pcnum 7   $ & $ \aret $ \\
                                                      &&&                & \\
  \end{tabular}
  \begin{tabular}{lll}

    \multicolumn{3}{c}{Effect of successive $\fetch$ directives} \\
    $\n$                        & $\buf$                                                 & $\rets$ \\\hline
    $\pcnum 3 \to \pcnum 5$     & $ \idxnum 3 \mapsto \icall $       & $\idxnum 3 \mapsto \push\ \pcnum 4$\\
                                & $ \idxnum 4 \mapsto \iopx{\mathop{succ}}{\rsp}{\rsp} $ & \\
                                & $ \idxnum 5 \mapsto \istore{\mathtt{4}}{[\rsp]} $      & \\
    $\pcnum 5 \to \pcnum 6$     & $ \idxnum 6 \mapsto \iopx{\addr}{r_d}{[12,r_b]} $      & \\
    $\pcnum 6 \to \pcnum 7$     & $ \idxnum 7 \mapsto \istore{r_d}{[\rsp]} $             & \\
    $\pcnum 7 \to \pcnum 4$     & $ \idxnum 8 \mapsto \iret $                            & $\idxnum 8 \mapsto \pop$\\
                                & $ \idxnum 9 \mapsto \iloadx{\rret}{[\rsp]} $           & \\
                                & $ \idxnum{10} \mapsto \iopx{\mathop{pred}}{\rsp}{\rsp} $ & \\
                                & $ \idxnum{11} \mapsto \ijmpi{[\rret]}{\pcnum 4} $        & \\
    $\pcnum 4 \to \pcnum 4$     & $ \idxnum{12} \mapsto \ilfence $                       & \\

  \end{tabular}
  \vspace{0.5em}\\
  \begin{tabular}{lll}

    Directive                              & Effect on $\buf$                                & Leakage \\\hline
    $\execute{\idxnum 4} $                 & $ \idxnum 4 \mapsto \ivalx{\rsp}{\mathtt{7B}} $ & \\
    $\execute{\idxnum 6} $                 & $ \idxnum 6 \mapsto \ivalx{r_d}{20} $           & \\
    $\executevalue{\idxnum 7} $            & $ \idxnum 7 \mapsto \istore{20}{[\rsp]} $       & \\
    $\executeaddr{\idxnum 7} $             & $ \idxnum 7 \mapsto \istore{20}{\mathtt{7B}} $  & $ \ofwd{\mathtt{7B}} $ \\
    $\execute{\idxnum 9} $                 & $ \idxnum 9 \mapsto \ivalx{\rret}{20} $         & $ \ofwd{\mathtt{7B}} $ \\
    \multirow{2}{*}{$\execute{\idxnum{11}}$} & $ \idxnum{12} \notin \buf $                     & $ \omiss, $ \\
                                           & $ \idxnum{11} \mapsto \ijmp{\pcnum{20}} $         & $ \ojump{\pcnum{20}} $ \\

  \end{tabular}

  \caption{Example demonstrating ``retpoline'' mitigation against Spectre v2 attack.
  The program is able to jump to program point $12 + r_b = \pcnum{20}$ without the schedule influencing prediction.}
  \label{fig:retpoline}
\end{figure}

\mypara{Examples}
We present an example of an RSB underflow attack in
\Cref{fig:rsb-attack}.
After fetching a $\call$ and paired $\ret$ instruction, the RSB will
be ``empty''.
When one more (unmatched) $\ret$ instruction is fetched, since
$\sttop(\rets) = \bot$, the program point $\n$ is no longer set by the
RSB, and is instead set by the (attacker-controlled) schedule.

\mypara{Retpoline mitigation}
A mitigation for Spectre v2 attacks is to replace indirect jumps with \emph{retpolines}~\cite{turner}.
\Cref{fig:retpoline} shows a retpoline construction that would replace the indirect jump in \Cref{fig:v2-attack}.
The $\call$ sends execution to program point $\pcnum 5$, while adding $\pcnum 4$ to the RSB.
The next two instructions at $\pcnum 5$ and $\pcnum 6$ calculate the same target as the indirect jump in \Cref{fig:v2-attack}
and overwrite the return address in memory with this jump target.
When executed speculatively, the $\ret$ at $\pcnum 7$ will pop the top value off the RSB, $\pcnum 4$, and jump there,
landing on a $\lfence$ instruction that loops back on itself. Thus speculative execution cannot proceed beyond this point.
When the transient instructions in the $\ret$ sequence finally execute, the indirect jump target $\pcnum{20}$
is loaded from memory, causing a roll back. However, execution is then directed to the proper jump target.
Notably, at no point is an attacker able to hijack the jump target via misprediction.

\section{Full proofs}
\label{sec:full-proofs}

\subsection{Consistency}

\begin{lemma}[Determinism]
If $C \step{d}{o'} C'$ and $C \step{d}{o''} C''$ then $C' = C''$ and $o' = o''$.
\end{lemma}
\begin{proof}
The tuple $(C, d)$ fully determines which rule of the semantics can be executed.
\end{proof}

\begin{definition}[Initial/terminal configuration]
  A configuration $C$ is an \emph{initial} (or \emph{terminal}) configuration if $|C.\buf| = 0$.
\end{definition}

\begin{definition}[Sequential schedule]
  Given a configuration $C$, we say a schedule $D$ is \emph{sequential}
  if every instruction that is fetched is executed and retired before further
  instructions are fetched.
\end{definition}

\begin{definition}[Sequential execution]
  $C \bigstep{D}{O} C'$ is a sequential execution if $C$ is an initial
  configuration, $D$ is a sequential schedule for $C$, and $C'$ is a terminal configuration.
\end{definition}
We write $C \bigstep{\seq}{O} C'$ if we execute sequentially.

\begin{lemma}[Sequential equivalence]
\label{lem:seq-correctness}
If $C \bigstep{D_1}{O_1} C_1$ is sequential and $C \bigstep{D_2}{O_2} C_2$ is sequential, then $C_1 = C_2$.
\end{lemma}
\begin{proof}
  Suppose $N = 0$. Then neither $D_1$ nor $D_2$ may contain any $\retire$ directives.
  Since we assume that both $C_1.\buf$ and $C_2.\buf$ have size 0,
  neither $D_1$ nor $D_2$ may contain any $\fetch$ directives.
  Therefore, both $D_1$ and $D_2$ are empty; both $C_1$ and $C_2$ are equal to $C$.

  We proceed by induction on $N$.

  Let $D_1'$ be a sequential prefix of $D_1$ up to the $N-1$th $\retire$,
  and let $D_1''$ be the remainder of $D_1$.
  That is, $\#\{d \in D_1' \mathop{|} d = \retire\} = N - 1$
  and $D_1' \| D_1'' = D_1$.
  Let $D_2'$ and $D_2''$ be similarly defined.

  By our induction hypothesis, we know $C \bigstep[N-1]{D_1'}{O_1'} C'$
  and $C \bigstep[N-1]{D_2'}{O_2'} C'$ for some $C'$.
  Since $D_1'$ (resp.~$D_2'$) is sequential and $|C'.\buf| = 0$, the first directive in $D_1''$ (resp.~$D_2''$)
  must be a $\fetch$ directive.
  Furthermore, $C' \bigstep[1]{D_1''}{O_1''} C_1$ and $C' \bigstep[1]{D_2''}{O_2''} C_2$.

  We can now proceed by cases on $C'.\mem[C'.\n]$, the final instruction to be fetched.

  \begin{itemize}
    \item For $\op$, the only valid sequence of directives is ($\fetch$, $\execute i$, $\retire$) where $i$ is
  the sole valid index in the buffer. Similarly for $\lfence$, with the sequence $\{\fetch, \retire\}$.

\item For $\load$, alias prediction is not possible, as no prior stores exist in the buffer.
  Therefore, just as with $\op$, the only valid sequence of directives is ($\fetch$, $\execute i$, $\retire$).

\item For $\store$, the only possible difference between $D_1''$ and $D_2''$ is the ordering
  of the $\execute{i:\,\textsf{value}}$ and $\execute{i:\,\textsf{addr}}$ directives. However, both
  orderings will result in the same configuration since they independently resolve the components of the $\store$.

\item For $\cond$, $D_1''$ and $D_2''$ may have different guesses for their initial fetch directives.
  However, both \textsc{cond-execute-correct} and \textsc{cond-execute-incorrect} will result
  in the same configuration regardless of the initial guess,
  as the $\cond$ is the only instruction in the buffer.
  Similarly for $\jmpi$.

\item For $\call$ and $\ret$, the ordering of execution of the resulting transient
  instructions does not affect the final configuration.
  \end{itemize}

  Thus for all cases we have $C_1 = C_2$.
\end{proof}

To make our discussion easier, we will say that a directive $d$ \emph{applies to}
a buffer index $i$ if when executing a step $C \step{d}{o} C'$:
\begin{itemize}
  \item $d$ is a $\fetch$ directive, and would fetch an instruction into index $i$ in $\buf$.
  \item $d$ is an $\executex$ directive, and would execute the instruction at index $i$ in $\buf$.
  \item $d$ is a $\retire$ directive, and would retire the instruction at index $i$ in $\buf$.
\end{itemize}

We would like to reason about schedules that do not contain \emph{misspeculated steps}, \ie
directives that are superfluous due to their effects getting wiped away by rollbacks.
\begin{definition}[Misspeculated steps]\label{def:misspeculatedsteps}
  Given an execution $C \bigstep{D}{O} C'$, we say that $D$ contains
  \emph{misspeculated steps} if there exists $d \in D$ such that $D' = D \setminus d$
  and $C \bigstep{D'}{O'} C'' = C'$.
\end{definition}
Given an execution $C \bigstep{D}{O} C'$ that may contain rollbacks, we
can create an alternate schedule $D^*$ without any rollbacks by removing all misspeculated steps.
Note that sequential schedules have no misspeculated steps\footnote{Sequential schedules may still misspeculate on conditional branches but  the rollback does not imply removal of any reorder buffer instructions as defined in Definition ~\ref{def:misspeculatedsteps}. } as defined in Definition~\ref{def:misspeculatedsteps}.

\begin{theorem}[Equivalence to sequential execution]
\label{thm:seq-correctness-full}
Let  $C$  be an initial configuration and $D$ a well-formed schedule for $C$.
If $C \bigstep{D}{O_1} C_1$, then $C \bigstep{\seq}{O_2} C_2$ and $C_1 \approx C_2$.
Furthermore, if $C_1$ is terminal then $C_1 = C_2$.
\end{theorem}
\begin{proof}
  Since we can always remove all misspeculated steps from
  any well-formed execution without affecting the final configuration, we
  assume $D_1$ has no misspeculated steps.

  Suppose $N = 0$. Then the theorem is trivially true.
  We proceed by induction on $N$.

  Let $D_1'$ be the subsequence of $D_1$ containing the first $N-1$ $\retire$ directives
  and the directives that apply to the same indices of the first $N-1$ $\retire$ directives.
  Let $D_1''$ be the complement of $D_1'$ with respect to $D_1$.
  All directives in $D_1''$ apply to indices later than any directive in $D_1'$, and thus cannot
  affect the execution of directives in $D_1'$. Thus $D_1'$ is a well-formed schedule
  and produces execution $C \bigstep[N-1]{D_1'}{O_1'} C_1'$.

  Since $D_1$ contains no misspeculated steps, the directives in $D_1''$
  can be reordered after the directives in $D_1'$. Thus $D_1''$ is a well-formed
  schedule for $C_1'$, producing execution $C_1' \bigstep[1]{D_1''}{O_1''} C_1''$
  with $C_1'' \approx C_1$.
  If $C_1$ is terminal, then $C_1''$ is also terminal and $C_1'' = C_1$.

  By our induction hypothesis, we know there exists $D_{\seq}'$ such that $C \bigstep[N-1]{D_{\seq}'}{O_2'} C_2'$.
  Since $D_1'$ contains equal numbers of $\fetch$ and $\retire$ directives, ends with a $\retire$, and contains no
  misspeculated steps, $C_1'$ is terminal.
  Thus $C_1' = C_2'$.

  Let $D_{\seq}''$ be the subsequence of $D_1''$ containing the $\retire$ directive in $D_1''$ and the directives that
  apply to the same index. $D_{\seq}''$ is sequential with
  respect to $C_1'$ and produces execution $C_1' \bigstep[1]{D_{\seq}''}{O_2''} C_2''$
  with $C_2'' \approx C_1'' \approx C_1$.
  If $C_1''$ is terminal, then $D_{\seq}'' = D_1''$ and thus $C_2'' = C_1'' = C_1$.

  Let $D_{\seq} = D_{\seq}' \| D_{\seq}''$. $D_{\seq}$ is thus itself sequential
  and produces execution $C \bigstep{\seq}{(O_2' \| O_2'')} C_2''$, completing our proof.
\end{proof}

\begin{corollary}[General consistency]
\label{thm:correctness}
Let  $C$  be an initial configuration.
If $C \bigstep{D_1}{O_1} C_1$ and $C \bigstep{D_2}{O_2} C_2$, then $C_1 \approx C_2$.
Furthermore, if $C_1$ and $C_2$ are both terminal then $C_1 = C_2$.
\end{corollary}
\begin{proof}
  By \Cref{thm:seq-correctness-full}, there exists $D_{seq}'$ such that executing
  with $C$ produces $C_1' \approx C_1$ (resp.~$C_1' = C_1$).
  Similarly, there exists $D_{seq}''$ that produces $C_2' \approx C_2$ (resp.~$C_2' = C_2$).
  By \Cref{lem:seq-correctness}, we have $C_1' = C_2'$. Thus $C_1 \approx C_2$ (resp.~$C_1 = C_2$).
\end{proof}

\subsection{Security}
\begin{restatable}[Label stability]{theorem}{thmLabelStability}
\label{thm:label-stability}
Let $\l$ be a label in the lattice $\Labels$.
If $C \bigstep{D_1}{O_1} C_1$ and $\forall o \in O_1: \l \notin o$,
then $C \bigstep{seq}{O_2} C_2$ and $\forall o \in O_2: \l \notin o$.
\end{restatable}
\begin{proof}
  Let $D_1^*$ be the schedule given by removing all misspeculated steps from $D_1$.
  The corresponding trace $O_1^*$ is a subsequence of $O_1$, and hence $\forall o \in O_1^*: \l \notin o$.
  We thus proceed assuming that execution of $D_1$ contains no misspeculated steps.

  Our proof closely follows that of \Cref{thm:seq-correctness-full}.
  When constructing $D_1'$ and $D_1''$ from $D_1$ in the inductive step, we know that
  all directives in $D_1''$ apply to indices later than any directive in $D_1'$, and cannot affect
  execution of any directive in $D_1'$. This implies that $O_1'$ is the subsequence of $O_1$
  that corresponds to the mapping of $D_1'$ to $D_1$.

  Reordering the directives in $D_1''$ after $D_1'$ do not affect the observations produced by
  most directives. The exceptions to this are $\executex$ directives for $\load$ instructions
  that would have received a forwarded value: after reordering, the $\store$ instruction they
  forwarded from may have been retired, and they must fetch their value from memory.
  However, even in this case, the address $a_{\l_a}$ attached to the observation does not change.
  Thus $\forall o \in O_2'': \l \notin o$.

  Continuing the proof as in \Cref{thm:seq-correctness-full}, we create schedule $D'_{\seq}$
  (with trace $O_2'$)
  from the induction hypothesis and $D''_{\seq}$
  (with trace $O_2''$)
  as the subsequence of $D_1''$ of directives
  applying to the remaining instruction to be retired. As noted before, executing the
  subsequence of a schedule produces the corresponding subsequence of the
  original trace; hence $\forall o \in O_2'': \l \notin o$.

  The trace of the final (sequential) schedule $D_{\seq} = D'_{\seq} \| D''_{\seq}$ is
  $O_2' \| O_2''$. Since $O_2'$ satisfies the label stability property via the induction hypothesis,
  we have $\forall o \in O_2' \| O_2'': \l \notin o$.

\end{proof}

By letting $\l$ be the label $\secret$, we get the following corollary:
\begin{corollary}[Secrecy]
\label{cor:label-secrecy}
If speculative execution of $C$ under schedule $D$ produces a trace $O$ that contains no $\secret$ labels,
then sequential execution of $C$ will never produce a trace that contains any $\secret$ labels.
\end{corollary}

With this, we can prove the following proposition:
\begin{proposition}
  For a given initial configuration $C$ and well-formed schedule $D$,
  if $C$ is \SONI with respect to $D$, and execution of $C$ with $D$ results in
  a terminal configuration $C_1$, then $C$ is also sequentially constant-time.
\end{proposition}
\begin{proof}
  Since $C$ is \SONI, we know that for all $C' \loweqv C$, we have $ C \bigstep{D}{O} C_1 $ and $ C' \bigstep{D}{O'} C_1' $
  where $C_1 \loweqv C_1'$ and $O = O'$. By \Cref{thm:seq-correctness-full}, we know there exist
  sequential executions such that $C \bigstep{\seq}{O_\seq} C_2$ and $C' \bigstep{\seq}{O'_\seq} C'_2$. Note that the
  two sequential schedules need not be the same.

  $C_1$ is terminal by hypothesis. Execution of $C'$ uses the same schedule
  $D$, so $C_1'$ is also terminal. Since we have $C_1 = C_2$ and $C_1' = C_2'$, we
  can lift $C_1 \loweqv C_1'$ to get $C_2 \loweqv C_2'$.

  To prove the trace property $O_\seq = O'_\seq$, we note that if $O_\seq \neq O'_\seq$, then
  since $C_2 \loweqv C_2'$,
  it must be the case that there exists some $o \in O_\seq$ such that $\secret \in O_\seq$.
  Since this is also true for $O$ and $O'$, we know that there exist no observations in either $O$ or $O'$
  that contain $\secret$ labels. By \Cref{cor:label-secrecy}, it follows that no $\secret$ labels appear
  in either $O_\seq$ or $O'_\seq$, and thus $O_\seq = O'_\seq$.
\end{proof}

\subsection{Soundness of \tool}
\label{sec:tool-soundness}

\begin{definition}[Affecting an index]
  We say a directive $d$ \emph{affects an index} $i$ if:
  \begin{itemize}
    \item $d$ is a $\fetch$-type directive and would produce a new mapping in $\buf$ at index $i$.
    \item $d$ is an $\executex$-type directive and specifies index $i$ directly (\eg $\executei$).
    \item $d$ is a $\retire$ directive and would cause the instruction at $i$ in $\buf$ to be removed.
  \end{itemize}
\end{definition}

\begin{definition}[Path function]
  The function $Path(C, D)$ produces the sequence of branch choice (from fetching $\cond$ instructions)
  and store-forwarding information (when executing $\load$ instructions)
  when executing $D$ with initial configuration $C$.
  That is, for a schedule $D$ without misspeculated steps:
  \begin{align*}
    Path(C, \emptyset) &= [] \\
    Path(C, D\|d) &= \begin{cases}
      Path(C, D) ; (i,b) ,& d = \fetchg{b} \\
      Path(C, D) ; (i,j) ,& d \textnormal{ produces }v_\l\{j,\a\} \\
      Path(C, D) ; (i,\bot) ,& d \textnormal{ produces }v_\l\{\bot,\a\} \\
      Path(C, D) ,& \textnormal{otherwise}
    \end{cases}
  \end{align*}
  where $d$ affects index $i$.
  If $D$ has misspeculated steps, then $Path(C, D) = Path(C, D^*)$
  where $D^*$ is the subset of $D$ with misspeculated steps removed.
  We write simply $Path(D)$ when $C$ is obvious.
\end{definition}

For the \Cref{lem:tool-fetch,lem:tool-execute,lem:tool-retire}, we start with the following shared assumptions:
\begin{itemize}
  \item $C$ is an initial configuration.
  \item $D_1$ and $D_2$ are nonempty schedules.
  \item $C \bigstepnon{D_1}{O_1} C_1$ and $C \bigstepnon{D_2}{O_2} C_2$.
  \item $Path(C, D_1) = Path(C, D_2)$.
  \item $D_1 = D_1'\|d_1$ and $D_2 = D_2'\|d_2$ and $d_1 = d_2$.
  \item $d_1$ and $d_2$ affect the same index $i$ in the their respective reorder buffers.
\end{itemize}
Let $o_1$ (resp.~$o_2$) be the observation produced during execution of $d_1$ (resp.~$d_2$).

\begin{lemma}[Fetch]
  \label{lem:tool-fetch}
  If $d_1$ and $d_2$ are both $\fetch$-type directives, then $C_1.\n = C_2.\n$ and $C_1.\buf[i] = C_2.\buf[i]$.
\end{lemma}
\begin{proof}
  Since fetches happen in-order, the index $i$ of a given physical instruction
  along a control flow path is deterministic.
  Both $D_1$ and $D_2$ both have the same (control flow) path.
  Since by hypothesis both $d_1$ and $d_2$ affect the same index $i$, $d_1$ and $d_2$ must necessarily both be fetching
  the same physical instruction.
  Furthermore, since $Path(D_1) = Path(D_2)$, if the fetched instruction is a $\cond$ instruction, then both $d_1$ and $d_2$ must have
  made the same guess.
  The lemma statements all hold accordingly.
\end{proof}

\begin{corollary}
  \label{lem:same-phys}
  If $D_1^*$ and $D_2^*$ are nonempty schedules such that
  $C \bigstepnon{D_1^*}{} C_1^*$ and $C \bigstepnon{D_2^*}{} C_2^*$ and
  $Path(C, D_1^*) = Path(C, D_2^*)$, then:
  For any $i \in C_1^*.\buf$, if $i \in C_2^*.\buf$, then
  both $C_1^*.\buf[i]$ and $C_2^*.\buf[i]$ were derived from the same physical instruction.
\end{corollary}
\begin{proof}
  Let $D_1$ be the prefix of $D_1^*$ such that the final directive in $D_1$ is the latest $\fetch$
  that affects $i$. Let $D_2$ be similarly defined w.r.t. $D_2^*$. Then by \Cref{lem:tool-fetch},
  $D_1$ and $D_2$ both fetch the same physical instruction to index $i$.
\end{proof}

\begin{lemma}
  \label{lem:tool-execute}
  If $d_1$ and $d_2$ are both $\executex$-type directives, then $C_1.\buf[i] = C_2.\buf[i]$ and $o_1 = o_2$.
\end{lemma}
\begin{proof}
  We proceed by full induction on the size of $D_1$.

  For the base case: if $\abs{D_1} = 1$, then the lemma statements are trivial regardless of the directive $d_1$.

  We know from \Cref{lem:same-phys} that since $d_1$ and $d_2$ both affect the same index $i$,
  the two transient instruction must be derived from the same physical instruction, and thus has the same register dependencies.
  For each register dependency $r$, if the register was calculated by a transient instruction at a prior index $j$,
  we can create prefixes $D_{1,j}$ and $D_{2,j}$ of $D_1$ and $D_2$ respectively that end at the $\executex$ directive
  that resolves $r$ at buffer index $j$. By our induction hypothesis, both $D_{1,j}$ and $D_{2,j}$ calculate the same value $v_\l$ for $r$.

  We now proceed by cases on the transient instruction being executed.

  \mypara{Op, Store (value)}
  Since all dependencies calculate the same values, both instructions calculate the same value.

  \mypara{Store (address)}
  Both instructions calculate the same address. Since $Path(D_1) = Path(D_2)$, both schedules have the same pattern of store-forwarding behavior.
  Thus execution of $d_1$ causes a hazard if and only if $d_2$ causes a hazard.

  \mypara{Load}
  Both instructions calculate the same address, producing the same observations $o_1$ and $o_2$.
  Since $Path(D_1) = Path(D_2)$, either $d_1$ and $d_2$ cause the values to be retrieved from
  the same prior $\store$s, or they both load values from the same address in memory.
  By our induction hypothesis, these values will be the same, so both instructions will resolve to the same value.

  \mypara{Branch}
  Both instructions calculate the same branch condition, producing the same observations $o_1$ and $o_2$.
  Since $Path(D_1) = Path(D_2)$, execution of $d_1$ causes a misspeculation hazard if and only if $d_2$ also causes misspeculation hazard.
\end{proof}

\begin{lemma}
  \label{lem:tool-retire}
  If $d_1$ and $d_2$ are both $\retire$ directives, then $o_1 = o_2$.
\end{lemma}
\begin{proof}
  From \Cref{lem:tool-fetch,lem:tool-execute} we know that for both $d_1$ and $d_2$, the transient instructions
  to be retired are the same. Thus the produced observations $o_1$ and $o_2$ are also the same.
\end{proof}

We now formally define the set of schedules examined by \tool:

\begin{definition}[Tool schedules]
  \label{def:tool-schedules}
  Given an initial configuration $C$ and a speculative window size $n$,
  we define the set of \emph{tool schedules} $D_T(n)$ recursively as follows:
  The empty schedule $\emptyset$ is in $D_T(n)$.
  If $D_0 \in D_T(n)$ and $C \bigstepnon{D_0}{} C_0$ and $\abs{C_0.\buf} < n$, then
  based on the next instruction to be fetched (and where $i$ is the index of the fetched instruction):
  \begin{itemize}
    \item $\op$: $D_0\|\fetch; \executei \in D_T(n)$.
    \item $\load$: $D_0\|\fetch; \executei \in D_T(n)$.
    \item $\store$: $D_0\|\fetch; \executeivalue \in D_T(n)$ and \newline
                    $D_0\|\fetch; \executeivalue; \executeiaddr \in D_T(n)$.
    \item $\cond$: Let $b$ be the ``correct'' path for the branch condition. Then
                   $D_0\|\fetchg{b}; \executei \in D_T(n)$ and \newline
                   $D_0\|\fetchg{\lnot b} \in D_T(n)$.
  \end{itemize}
  Otherwise, if $\abs{C_0.\buf} = n$, then we instead extend based on the oldest instruction in the reorder buffer.
  If the oldest instruction is a $\store$ with an unresolved address, and will not cause a hazard, then $D_0\|\executeiaddr; \retire \in D_T(n)$.
  Otherwise, if the oldest instruction is fully resolved, then $D_0\|\retire \in D_T(n)$.
\end{definition}

\begin{proposition}[Path coverage]
  \label{prop:path-coverage}
  If $D_1$ is a well-formed schedule for $C$ whose reorder buffer never grows beyond size $n$, then
  $\exists D_2 : Path(D_1) = Path(D_2) \land D_2 \in D_T(n)$.
\end{proposition}
\begin{proof}
  The proof stems directly from the definition of $D_T(n)$; at every branch, both branches are added to the set of schedules,
  and every load is able to ``skip'' any combination of prior stores.
\end{proof}

\begin{theorem}[Soundness of tool]
\label{thm:tool-soundness}
  If speculative execution of $C$ under a schedule $D$ with speculation bound $n$ produces a trace $O$ that contains at least one secret label,
  then there exists a schedule $D_t \in D_T(n)$ that produces a trace $O_t$ that also contains at least one secret label.
\end{theorem}
\begin{proof}
  We can truncate $D$ to a schedule $D^*$ that ends at the first directive to produce a secret observation.
  By \Cref{prop:path-coverage} there exists a schedule $D_0 \in D_T(n)$ such that $Path(D_t) = Path(D^*)$. By following construction
  of tool schedules as given in \Cref{def:tool-schedules}, we can find a
  schedule $D_t \in D_T(n)$ that satisfies the preconditions for \Cref{lem:tool-execute}.
  Then by that same lemma, $D_t$ produces the same final observation as $D^*$, which contains a secret label.
\end{proof}

\end{document}